\documentclass[journal]{IEEEtran}

\usepackage{setspace,cite}
\usepackage{graphicx}
\usepackage{epstopdf}
\usepackage{caption}
\usepackage{epsfig} 
\usepackage{amsmath} 
\usepackage{amssymb} 
\usepackage{amsthm}
\usepackage[dvipsnames]{xcolor}
\usepackage{subcaption}
\usepackage{verbatim}  
\usepackage{multirow}
\newtheorem{theorem}{Theorem}
\newtheorem{lemma}{Lemma}
\newtheorem{proposition}{Proposition}
\newtheorem{corollary}{Corollary}
\newtheorem{remark}{Remark}

\usepackage[ colorlinks = true,
linkcolor = blue,
urlcolor = blue,
citecolor = red,
anchorcolor = green,
]{hyperref}

\newcommand{\Dalpha}{D_{\alpha}}
\newcommand{\bp}{\mathbf{p}}
\newcommand{\bW}{\mathbf{W}}
\newcommand{\bw}{\mathbf{w}}
\newcommand{\bA}{\mathbf{A}}

\newcommand{\bv}{\mathbf{v}}
\newcommand{\bB}{\mathbf{B}}
\newcommand{\ba}{\mathbf{a}}
\newcommand{\blambda}{\boldsymbol{\lambda}}

\newcommand{\bV}{\mathbf{V}}
\newcommand{\bU}{\mathbf{U}}

\newcommand{\bSigma}{\boldsymbol{\Sigma}}
\allowdisplaybreaks[2]
\DeclareMathOperator{\Tr}{Tr}

\author{Jiachun Liao, Lalitha Sankar,  Vincent Y. F. Tan and Flavio du Pin Calmon
\thanks{This work is supported in part by the National Science Foundation under grants CCF\--1350914 and CIF\--1422358.}
}
\begin{document}

\title{Hypothesis Testing under Mutual Information Privacy Constraints in the High Privacy Regime}

\maketitle



\begin{abstract}
Hypothesis testing  is a statistical inference framework for determining the true distribution among a set of possible distributions for a given dataset. Privacy restrictions may require the curator of the data or the respondents themselves to share data with the test only after applying a randomizing \textit{privacy mechanism}. This work considers mutual information (MI) as the privacy metric for measuring leakage. In addition, motivated by the Chernoff-Stein lemma, the relative entropy between pairs of distributions of the output (generated by the privacy mechanism) is chosen as the utility metric. For these metrics, the goal is to find the optimal privacy-utility trade-off (PUT) and the corresponding optimal privacy mechanism for both binary and $m$-ary hypothesis testing. Focusing on the high privacy regime, Euclidean information-theoretic approximations of the binary and $m$-ary PUT problems are developed. The solutions for the approximation problems clarify that an MI-based privacy metric preserves the privacy of the source symbols in inverse proportion to their likelihoods.
\end{abstract}

\begin{IEEEkeywords}
Hypothesis testing, privacy-guaranteed data publishing, privacy mechanism, Euclidean information theory, relative entropy, R\'enyi divergence, mutual information.
\end{IEEEkeywords}
\vspace{-0.15cm}
\section{Introduction}
There is tremendous value to publishing datasets for a variety of statistical inference applications; however, it is crucial to ensure that the published dataset, while providing utility, does not reveal potentially privacy-threatening information. Specifically, the published dataset should allow the intended inference to be made while limiting other unwanted inferences. This requires using a randomizing mechanism (i.e., a \textit{noisy channel}) that guarantees a certain measure of privacy. Any such \textit{privacy mechanism} may, in turn, reduce the fidelity of the intended inference leading to a trade-off between utility of the published data and the privacy of the respondents in the dataset. 

We consider the problem of privacy-guaranteed data publishing hypothesis testing. The use of large datasets to test two or more hypotheses (e.g., the 99\%-1\% theory of income distribution in the United States~\cite{economist_article}) relies on the classical statistical inference framework of binary or multiple hypothesis testing. The optimal test for hypothesis testing under various scenarios (non-Bayesian,  Bayesian, minimax) involves the so-called Neyman-Pearson (or likelihood ratio) test~\cite{PoorBook} in which the likelihood ratio of the hypotheses is compared to a given threshold. We focus exclusively on the non-Bayesian setting. In particular, for $m$-ary ($m\geq 2$) hypothesis testing problem, we consider the setting in which the probability of missed detection is minimized for one specific hypothesis (e.g., presence of cancer) while requiring the probabilities of false alarm for the same hypothesis to be bounded (relative to the remaining hypotheses). In this context, we can apply the Chernoff-Stein Lemma~\cite[Chapter~11]{IT_Cover} which states that for a pair of hypotheses the largest error exponent of the missed detection probability, under the constraint that the false alarm probability is bounded above by a constant, is the relative entropy between the probability distributions for the two hypotheses.

Inspired by the Chernoff-Stein lemma, for the $m$-ary hypothesis setting described above, we use relative entropy as a measure of the utility of the published dataset (for hypothesis testing), and henceforth, refer to this as the \textit{relative entropy setting}. Furthermore, for binary hypothesis testing ($m=2$), we also consider the setting in which the probabilities of both missed detection and false alarm decrease exponentially. For this setting, using known results of hypothesis testing \cite{EEHT_Tuncel}, we take the R\'enyi divergence as the utility metric and refer to this as the \textit{R\'enyi divergence setting}. For the privacy metric, we use mutual information between the original and published datasets as a measure of the additional knowledge (privacy leakage) gained on average from the published dataset. By bounding the MI leakages, our goal is to develop privacy mechanisms that restrict the relative entropy between the prior and posterior (after publishing) distributions of the dataset, averaged over the published dataset. By restricting the distance between prior and posterior beliefs, we capture a large class of computationally unbounded adversaries that can use different inference methods. Specifically, bounding MI leakage allows us to exploit information-theoretic relationships between MI and the probabilities of detection/estimation to bound the ability of an adversary to ``learn'' the original dataset \cite{IT_Cover,du_pin_calmon_privacy_2012,RebolloMonedero_TClosenessLikePrivacy_2010}.

\subsection{Our Contributions}
We study the privacy-preserving data publishing problem by considering a \textit{local privacy} model in which the {\em same} (memoryless) mechanism is applied independently to each entry of the dataset. This allows the respondents of a dataset to apply the privacy mechanism {\em before} sharing data. Our main contributions are as follows:

\begin{enumerate}
\item  We introduce the privacy-utility trade-off (PUT) problem for hypothesis testing (Section \ref{section:Problem_Formulation}). The resulting PUT involves maximizing the minimum of a set of relative entropies, subject to constraints on the MI-~based leakages for all source classes. 

\item The PUT problem involves maximizing the minimum of a set of convex functions over a convex set which is, in general, NP-hard. In Section \ref{section:approximation}, we approximate the trade-off in the high privacy regime (near zero leakage) using techniques from Euclidean information theory (E-IT); these techniques have found use in deriving capacity results in~\cite{EITzheng2008,EIT2015}.

\item For binary hypothesis testing, we first consider the relative entropy setting (Section \ref{subsection:BHT_relative}), in which we determine the optimal mechanism in closed form for the E-IT approximation by exploring the problem structure. Our results suggest that the solution to the E-IT approximation is independent of the alphabet size and, more importantly, that a MI-based privacy metric preserves the privacy of the source symbols inversely proportional to their likelihoods, thereby, providing more distortion to the (informative) outliers in the dataset which, in general, are more vulnerable to detection.
 
 \item We extend our analysis to the R\'enyi divergence setting (Section \ref{subsection:BHT_renyi}), the optimal mechanism for its E-IT approximation problem in high privacy regime is similar in form to the relative entropy setting.
 
 \item We study the $m$-ary hypothesis testing problem (Section \ref{section:MHT}) and show that optimal solutions of the E-IT approximation can be obtained via semidefinite programs (SDPs) \cite{boydconvex}. Specially, for binary sources the optimal mechanism is derived in closed form. The dependence on the source distribution is highlighted here as well.

 \item In Section \ref{section:illustration}, via numerical simulations, we uncover regimes of distribution  tuples and leakage values for which the E-IT approximation is accurate.
 \end{enumerate}

\subsection{Related Work}
Privacy-guaranteed hypothesis testing in the high privacy regime using MI as the privacy metric was first studied by the authors in \cite{Liao_Allerton16}. Specifically, the focus of \cite{Liao_Allerton16} is on the relative entropy setting of binary hypothesis test in the high privacy regime.
We significantly extend that work with three key contributions: (a) we derive optimal mechanisms in the high privacy regime for binary hypothesis test in the R\'enyi divergence setting; (b) we derive optimal mechanisms in the high privacy regime for the $m$-ary problem in relative entropy setting; (c) we provide detailed illustrations of results for binary  and $m$-ary hypothesis testing. 

Recently, the problem of designing privacy mechanisms for hypothesis testing has gained interest. Kairouz et al.~\cite{Kairouz2014} show that the optimal locally differential privacy (L-DP) mechanism has a \textit{staircase} form and can be obtained as a solution of a linear program. Gaboardi et al.~\cite{DP_ChiSquared_Gaboardi} deal with a privacy-guaranteed hypothesis testing by using chi-square goodness of fit as the utility measure and adding Gaussian or Laplace noise to dataset to guarantee DP-based privacy protection.

Our problem differs from these efforts in using MI  as the privacy metrics. In \cite{Kairouz2014}, the L-DP formulation, focused on the high privacy regime, requires the mechanism to limit distinction between any two letters of the source alphabet for a given output. The requirement also gathers all privacy mechanisms satisfying a desired privacy protection measured by L-DP within a hypercube. Therefore, the authors simplify the trade-off problem to a linear program by exploring the sub-linearity of the relative entropy function. In contrast, all privacy mechanisms giving a desired MI-based privacy form a  convex set which is not a polytope. However, taking advantage of E-IT, we propose good approximations for the MI-based privacy utility trade-offs in high privacy regime. In fact, we present closed-form privacy mechanisms for both binary hypothesis testing with arbitrary alphabets as well as $m$-ary hypothesis testing with binary alphabets. Furthermore, for $m$-ary hypothesis testing with arbitrary sources, the privacy mechanism can be attained effectively by solving an SDP.

The connection between hypothesis testing and privacy has been studied in the context of location anonymization and smart meter privacy. In location privacy, the problem of determining if a sequence of anonymized data points (e.g. location positions without an accompanying user ID) belongs to a target user can be formulated as a hypothesis test. More specifically, if the distribution of the user's data is known and unique among other users, any observed sequence can be tested against the hypothesis that it was drawn from this distribution, thus revealing if it belongs to the target user. Within this context, Montazeri \text{et al.}  \cite{montazeri_defining_2016,montazeri_achieving_2016}  studied the problem of anonymizing sequences of location data, and characterized the probability of correctly guessing a target user's data within a larger dataset. In related work on smart meter privacy, Li and Oechtering \cite{li_privacy_2015}  considered the problem of private information leakage in a smart grid. Here, an adversary challenges a consumer's privacy by performing an unauthorized binary hypothesis test on the consumer's behavior based on smart meter readings. Li and Oechtering \cite{li_privacy_2015} propose a solution for mitigating the incurred privacy risk with the assist of an alternative energy source. 

The theoretical analysis done by Montazeri \text{et  al.}  \cite{montazeri_defining_2016,montazeri_achieving_2016}   and Li and Oechtering  \cite{li_privacy_2015}  are related to the one presented here in that they also make use of large deviation (information-theoretic) results in hypothesis testing. However, we apply these powerful theoretical tools to a different setting, in which data is purposefully randomized before disclosure in order to provide privacy, while guaranteeing utility in terms of a successful hypothesis test. Whereas they consider a hypothesis testing adversary, here we consider a precise hypothesis test as part of the utility metric.  

MI has been amply used as a measure for quantifying information leakage within the information-theoretic privacy literature (cf.~\cite{sankar_utility-privacy_2013,sankar_smart_2013,du_pin_calmon_privacy_2012,salamatian_how_2013,calmon_fundamental_2015,rebollo-monedero_t-closeness-like_2010,sankar_information-theoretic_2010} and the references therein). The connection between MI-based metrics and other privacy metrics has been studied, for example, by Makhdoumi and Fawaz \cite{makhdoumi_privacy-utility_2013}. In the present paper, we approximate MI by the chi-squared divergence which, in turn, also posses interesting estimation-theoretic properties \cite{calmon_bounds_2013}. An exploration of the role of   chi-squared related metrics in privacy has appeared in the work of Asoodeh et al. \cite{asoodeh_maximal_2015,asoodeh_privacy-aware_2016}.

\subsection{Notation}
We use bold capital letters to represent matrices, e.g., $\mathbf{X}$ is a matrix with the $i^{\mathrm{th}}$ row (or column)  being $\mathbf{X}_i$ and the $(i,j)^{\mathrm{th}}$ entry $X_{ij}$. We  use bold lower case letters to represent vectors, e.g. $\mathbf{x}$ is a vector with the $i^{\mathrm{th}}$ entry $x_i$. Sets are denoted by capital calligraphic letters. 

For vectors $\mathbf{a}$ and $\mathbf{b}$, and  functions $f$ and $g$,  $\big[\frac{f(\mathbf{a})}{g(\mathbf{b})}\big]$ is a diagonal matrix with the $i^{\mathrm{th}}$ diagonal entry being $\frac{f(a_i)}{g(b_i)}$, e.g., the diagonal matrix $[\frac{\mathbf{a}^2}{\sqrt{\mathbf{b}}}]$ has diagonal entries $\frac{a_i^2}{\sqrt{b_i}}$. We denote the $l_2$-norm of a vector $\mathbf{x}$ by $\|\mathbf{x}\|$, the logarithm of $x$ to the base $2$ as $\log x$.  Probability mass functions are denoted as row vectors, e.g., $\mathbf{p}$. In addition, $D(\cdot \|\cdot)$ denotes the relative entropy and $I(\cdot ;\cdot)$ denotes the MI.  We can write the MI between two random variables or between a probability distribution and the corresponding conditional probability matrix. Indeed, for two random variables $X,\bar{X}$ with $X\sim \mathbf{p}$ and $\bar{X}|\{X=x\}\sim \mathbf{W}_{\bar{X}|X=x}$, the MI is denoted as $I(X;\bar{X})$ or $I(\mathbf{p},\mathbf{W}_{\bar{X}|X})$.

\section{Problem Formulation}\label{section:Problem_Formulation}
\subsection{General Hypothesis Testing}
We consider an $m$-ary hypothesis testing problem that distinguishes between $m\ge 2$ explanations for an observed dataset.
Let $X^n=(X_1,\ldots,X_n)$ denote a sequence of $n$   random variables, where the entries $X_i$ are drawn independently according to a probability distribution $\mathbf{p}$. The observed random variables are assumed to be discrete with alphabet $\mathcal{X}$ and size $|\mathcal{X}|=M$. The $m$ hypotheses are denoted as $H_k:\, \bp=\bp_k$ for $k\in\{1,\ldots,m\}$. Our utility goal is to make a decision about the underlying distribution of the data $X^n$. Let the disjoint decision regions be $\mathcal{A}_k^{(n)}$. This means that if $X^n $ belongs to $\mathcal{A}_k^{(n)}$, we decide in favor of $H_k$.  

\subsection{Binary Hypothesis Testing}
In binary hypothesis testing, there are only two hypotheses: $H_1:\, \mathbf{p}=\mathbf{p}_1$ and $H_2:\, \mathbf{p}=\mathbf{p}_2$. The optimum test is the Neyman-Pearson test in which the decision region for hypothesis $H_1$ is  $\mathcal{A}_1^{(n)} =\big\{ \mathbf{x}^n \,:\,  \frac{\mathbf{p}_1(\mathbf{x}^n)}{\mathbf{p}_2(\mathbf{x}^n)}>T\big\}$ for some threshold $T\in\mathbb{R}$. Let $\beta_1^{(n)}$ and $\beta_2^{(n)}$ be the probabilities of false alarm and missed detection for $H_1$, respectively. Use $\beta_2^{(n)}(\delta)$ to indicate the smallest probability of the missed detection subject to the condition that $\beta_1^{(n)}\le\delta$. The Chernoff-Stein lemma~\cite[Chap.~11]{IT_Cover} states that
\begin{equation}
 \lim_{n\to\infty}-\frac{1}{n}\log\beta_2^{(n)}(\delta)=D(\bp_2\|\bp_1),\quad\forall\, \delta\in (0,1).  \label{eqn:stein}
 \end{equation} 
Hence, we use $D(\bp_2\|\bp_1)$ as our utility function.

\subsection{$m$-ary Hypothesis Testing}
In $m$-ary hypothesis testing, there are 
$m(m-1)$ different errors resulting from mistaking hypothesis $H_i$ for $H_j$, $i\ne j$. To keep our analysis simple, we consider a scenario somewhat analogous to the ``red alert'' \cite{RedAlert_Nazer} problem in ``unequal error protection'' \cite{UnequalErrorProtection_Borade}. There is one distinguished hypothesis $H_1$ whose inference takes precedence. For example, in practice $\bp_1$ could be the underlying distribution of measurements of a  malignant tumor; the other distributions $\bp_2,\ldots, \bp_k$ could    be the underlying distributions of  measurements of  various benign tumors. We would like to minimize the miss-detection rate of $H_1$. In this scenario, we would design the decision regions  $\{\mathcal{A}_k^{(n)}\}_{k=1}^m$ to maximize the minimum of $E_{1,k}$ over $k \in \{2,\ldots, m\}$, where $E_{ 1,k}$ is the error exponent  (exponential rate of decay of the error probability analogous to~\eqref{eqn:stein}) of mistaking $H_k$ when $H_1$ is true. 

 
\subsection{Privacy Considerations}
In most data collection and classification applications, there may be an additional requirement to ensure that the dataset, while providing utility, does not leak information about the respondents of the data. This in turn implies that the data provided to the hypothesis test is not the same as the original data, but instead a randomized version that guarantees precise measures of privacy (information leakage) and utility. Specifically, we use MI  as a measure of the average information leakage between the input dataset and its randomized output dataset that is used by the test. The goal is to find the randomizing mapping, henceforth referred to a \textit{privacy mechanism}, such that a measure of utility of the data is maximized while ensuring that the MI-based leakages for all possible source classes are bounded.

We assume that the entries of the dataset are generated in an i.i.d.\ fashion. 
Focusing on the local privacy model, the randomizing privacy mechanism for the hypothesis testing problem is memoryless. Let $\mathbf{W}$, an $M\times N$ conditional probability matrix, denote this memoryless privacy mechanism which maps the $M$ letters of the input alphabet $\mathcal{X}$ to $N$ letters of the output alphabet $\hat{\mathcal{X}}$, where $N\geq 2$ is an arbitrary finite integer. 
Thus, the i.i.d.\ sequence $X^n \sim \mathbf{p}_k, k\in \{1,\ldots,m\}$, is mapped to an output sequence $\hat{X}^n$ whose entries $\hat{X}_j \in \hat{\mathcal{X}}$ for all $j\in \{1,\ldots,N\}$ are i.i.d.\ with the distribution $\mathbf{p}_k\mathbf{W}$. Thus, the hypothesis test is now performed on a sequence $\hat{X}^n$ that belongs to one of $m$ source classes with distributions\footnote{We remind that the distribution $\mathbf{p}\mathbf{W}$ is the {\em output distribution} induced by the input (row vector) $\mathbf{p}$ and the privacy mechanism (transition matrix) $\mathbf{W}$.} $\mathbf{p}_k\mathbf{W}$. For the $m$-ary setting, the error exponent, corresponding to the missed detection of $H_1$ as $H_k$, is $D(\bp_k\bW\|\bp_1\bW)$.

\subsection{The Privacy-Utility Trade-off}
To design an appropriate privacy mechanism, we wish to maximize the minimum of the $m-1$ error exponents  $D(\bp_k\bW\|\bp_1\bW)$ subject to the following leakage constraints: $I(\mathbf{p}_k,\mathbf{W})\leq \epsilon_k$ for $k\in\{1,\ldots,m\}$. Formally, the privacy-utility trade-off (PUT) problem is that finding the optimal privacy mechanism $\mathbf{W}^*$ of the following optimization:
\begin{equation}\label{eq:MHT_original}
\begin{aligned}
\max_{\substack{\bW\in\mathcal{W}}} \quad & \min_{k=2,\ldots, m} D(\bp_k\bW\|\bp_1\bW)\\
\mathrm{s.t.} \quad & I(\bp_k,\bW)\leq \epsilon_k &k=1,2,\ldots,m&	
\end{aligned}
\end{equation}
where 
$\mathcal{W}$ is the set of   $M\times N$ row stochastic matrices, and $\epsilon_k \in [0, H(\mathbf{p}_k)]$, $k\in\{1,\ldots,m\}$, are permissible upper bounds on $I(\bp_k,\mathbf{W})$. The optimization in \eqref{eq:MHT_original} maximizes the minimum of $m-1$ convex functions over a convex set. Since the maximum of each of the $m-1$ convex functions are attained on the boundary of the feasible region, the optimal solution  of the optimization is also on the boundary. Because of the MI constraints, the feasible region is, in general, not a polytope, and thus, has infinitely many extremal points. While there exist computationally tractable methods to obtain a solution by approximating the feasible region by an intersection of polytopes \cite{MinConcave}, our focus is on developing a principled approximation for \eqref{eq:MHT_original} in a specific privacy regime to obtain a closed-from and easily-interpretable privacy mechanism. 

Specifically, we will work in the high privacy regime in which $\epsilon_k$ is small. In this regime, one can use Taylor series expansions to approximate both the objective function and the constraints. Such approximations  were   considered in \cite{ITSt_Tutorial, EITzheng2008}. More recently, analyses based on such approximations, referred to as E-IT, have been found to be useful in a variety settings from  graphical model learning~\cite{Tan11} to network information theory problems \cite{EITzheng2008}\cite{EIT2015}.

\section{Approximations in the High Privacy Regime}\label{section:approximation}
In this section, we develop E-IT approximations of the relative entropy $D(\bp_k\bW\|\bp_1\bW)$ and the MI $I(\bp_k,\bW)$ functions, based on which we propose an approximation of PUT in \eqref{eq:MHT_original} in the high privacy regime. 

To develop an approximation, we select an operating point which will be perturbed to provide an approximately-optimal privacy mechanism. We let $\epsilon_k\in [0,\epsilon^*] \text{ for all }k$ where  $\epsilon^* \ll \min\{H(\mathbf{p}_k),k\in\{1,\ldots,m\}\}$.  Since our focus is on the high privacy regime, we present the approximation around a perfect privacy operation point, i.e., a privacy mechanism $\mathbf{W}_0$ that achieves $\epsilon_k=0$ for all $k$.

\begin{lemma}\label{lemma:perfect_privacy}
	For perfect privacy, i.e., $\epsilon_k=0$ for all $k\in\{1,\ldots,m\}$, the privacy mechanism $\mathbf{W}_0$ is a rank-$1$ row stochastic matrix with every row being equal to a row vector $\mathbf{w}_0$ where $\mathbf{w}_0$ belongs to probability simplex, such that the entries $w_{0j},j\in \{1,\ldots,N\}$ of the vector $\mathbf{w}_0$ satisfy
	\begin{align}
	&\sum_{j=1}^{N}w_{0j}=1,\qquad \mathrm{and}\\
	&w_{0j} \geq 0 \quad \forall j\in \{1, \ldots, N\} .
	\end{align}
\end{lemma}
\begin{proof}
	For any probability distribution $\mathbf{p}$ with entries $p_i, \,i\in \{1,\ldots,M\}$, and a privacy mechanism $\mathbf{W}$, 
	\begin{align}
	I(\mathbf{p},\mathbf{W})
	&= \sum_{i=1}^{M}\sum_{j=1}^{N}p_iW_{ij}\log\frac{W_{ij}}{\sum_{i=1}^{M}p_iW_{ij}}\\
	\label{eq:perfectprivacy_logInequality}
	&	\geq  \sum_{i=1}^{M}p_i\bigg(\sum_{j=1}^{N}W_{ij}\bigg)\log\frac{\sum_{j=1}^{N}W_{ij}}{\sum_{j=1}^{N}\sum_{i=1}^{M}p_iW_{ij}}\\*
	&= \sum_{i=1}^{M}p_i\sum_{j=1}^{N}W_{ij}\log\frac{1}{1}=0
	\end{align}
	where \eqref{eq:perfectprivacy_logInequality} results from the log-sum inequality. Equality in \eqref{eq:perfectprivacy_logInequality} holds if and only if \cite[Theorem 2.7.1]{IT_Cover}
	\begin{align}
	\label{eq:perfectprivacy_W}
	W_{ij}=\sum_{k=1}^{M}p_kW_{kj},\;\; i\in \{1, \ldots, M\} ,j\in \{1, \ldots, N\}.
	\end{align}
	In other words, perfect privacy, i.e., zero leakage, is achieved when every row of the optimal mechanism $\mathbf{W}_0$ is the same and is equal to the probability distribution $\mathbf{w}_0=\mathbf{p}\mathbf{W}_0$. 
\end{proof}
Thus, for the perfect privacy setting, the optimal mechanism satisfying \eqref{eq:perfectprivacy_W} does not rely on the input distribution. 
\begin{remark}
	Note that, for any $\mathbf{W}_0$ satisfying \eqref{eq:perfectprivacy_W} that achieves perfect privacy, the utility is $D(\mathbf{p}_k\mathbf{W}_0 \| \mathbf{p}_1\mathbf{W}_0)=0$ for all $k\in\{2,\ldots,m\}$. Furthermore, the rows of $\mathbf{W}_0$, i.e., $\mathbf{w}_0$, can take any value in an $N$-dimensional probability simplex. 
\end{remark}

The following proposition presents a E-IT approximation for the objective and constraint functions of the optimization in \eqref{eq:MHT_original}, i.e., the relative entropy $D(\bp_k\bW\|\bp_1\bW)$ and MI $I(\bp_k,\bW)$. This approximation is only applicable to the high privacy regime in which the privacy mechanism $\mathbf{W}$ is modeled as a perturbation of a  $\mathbf{W}_0$ per Lemma \ref{lemma:perfect_privacy}.
\begin{proposition}\label{Proposition:approximation_MIandRE}
	In the high privacy regime with 
	$0\leq\epsilon_k \ll \min\{H(\mathbf{p}_k),k\in\{1,\ldots,m\}\}$, 
	the privacy mechanism $\mathbf{W}$ is chosen as a perturbation of a perfect privacy $(\epsilon_k=0 \text{ for all k})$ achieving mechanism $\mathbf{W}_0$, i.e., $\mathbf{W}=\mathbf{W}_0+\boldsymbol{\Theta}$. The mechanism $\mathbf{W}_0$ is a rank-1 row stochastic matrix with every row being equal to a row vector $\mathbf{w}_0$ whose entries $w_{0j}$ satisfy $\sum_{j=1}^{N}w_{0j}=1$ and $w_{0j}>0$, for all $j\in\{1,\ldots,N\}$. The perturbation matrix $\boldsymbol{\Theta}$ is an $M\times N$ matrix with entries $\Theta_{ij}$ satisfying $\sum_{j=1}^{N}\Theta_{ij}=0$ and $|\Theta_{ij}| \leq \rho w_{0j}$, for all $i\in \{1, \ldots, M\},j\in \{1, \ldots, N\}$. For this perturbation model, the relative entropy $D(\bp_k\bW\|\bp_1\bW)$ for all $k\in \{2, \ldots, m\}$ in the objective function, and the MI $I(\bp_k,\bW)$ for all $k\in \{1, \ldots, m\}$ in the constraints of \eqref{eq:MHT_original} can be approximated as
	\begin{align}
	\label{eq:EITapproximation_RE}
	D(\bp_k\bW\|\bp_1\bW)&\approx\frac{1}{2}\big\|(\mathbf{p}_k-\mathbf{p}_1)\boldsymbol{\Theta}[(\mathbf{w}_0)^{-\frac{1}{2}}]\big\|^2\\
	\label{eq:EITapproximation_MI}
   I(\bp_k,\bW)&\approx	\frac{1}{2}\sum_{i=1}^{M}p_{2i}\big\|\boldsymbol{\Theta}_i[(\mathbf{w}_0)^{-\frac{1}{2}}]\big\|^2
	\end{align}
	where $p_{ki}$, for $k\in\{1,\ldots,m\}$, is the $i^{\mathrm{th}}$ entry of $\mathbf{p}_k$, $\boldsymbol{\Theta}_i$ is the $i^{\mathrm{th}}$ row of $\boldsymbol{\Theta}$, and $[(\mathbf{w}_0)^{-\frac{1}{2}}]$ is a diagonal matrix with $i^{\mathrm{th}}$ diagonal entry, for all $i$, being $(w_{0i})^{-\frac{1}{2}}$. For ease of analysis, setting $\mathbf{A}=\boldsymbol{\Theta}[(\mathbf{w}_0)^{-\frac{1}{2}}]$, \eqref{eq:EITapproximation_RE} and \eqref{eq:EITapproximation_MI} can be rewritten as 
		\begin{align}
		\label{eq:EITapproximation_RE_A}
		D(\bp_k\bW\|\bp_1\bW) &\approx\frac{1}{2}\|(\mathbf{p}_k-\mathbf{p}_1)\mathbf{A}\|^2\\
		\label{eq:EITapproximation_MI_A}
		I(\bp_k,\bW)&\approx \frac{1}{2}\sum_{i=1}^{M}p_{ki}\|\mathbf{A}_i\|^2
		\end{align}
		In \eqref{eq:EITapproximation_RE} and~\eqref{eq:EITapproximation_RE_A}, the notation $\approx$ means that the difference between the left and right sides is $o( \| \mathbf{p}_k-\mathbf{p}_1 \|_\infty^2)$. Similarly,  in \eqref{eq:EITapproximation_MI} amd \eqref{eq:EITapproximation_MI_A}, $\approx$ means that the two sides differ by $o( \| \boldsymbol{\Theta} \|_\infty^2 )$.  
\end{proposition}
	Note that in Proposition \ref{Proposition:approximation_MIandRE}, $\bw_0$ is in the interior of the probability simplex, i.e., $\bw_0>0$. The approximation results from the observation that all rows of a privacy mechanism $\bW$ in the high privacy (low leakage) regime are very close to each other and both the relative entropy and MI can be approximated by the $\chi^2$ divergence. The detailed proof is in Appendix~\ref{proof:Proposition_approximation_MIandRE}.

\section{Binary Hypothesis Testing in the High Privacy Regime}\label{section:BHT}
For binary hypothesis testing, there  are only two hypotheses $H_1:\bp=\bp_1$ and $H_2:\bp=\bp_2$, and therefore, only two types of errors. In this section, we consider the simplest hypothesis testing scenario under two regimes. First, we regard one of the two hypotheses (e.g., $H_1$) as being more important than the other. In this case, the goal is to maximize the exponent of the missed detection for $H_1$ subject to an upper bound on its false alarm probability. Second, both hypotheses are important and the goal is to maximize a weighted sum of the two exponents of the false alarm and missed detection. For both cases, we derive the PUTs in the high privacy regime and provide  methods to attain explicit privacy mechanisms. 

\subsection{Binary Hypothesis Testing (Relative Entropy Setting)}\label{subsection:BHT_relative}
We consider the case in which the false alarm of $H_1$ is  bounded by a fixed positive constant and we examine the fastest rate of decay of its missed detection. This is exactly the problem formulated in Section \ref{section:Problem_Formulation}, and the PUT in~\eqref{eq:MHT_original} becomes
\begin{equation}\label{eq:OrigProHP}
\begin{aligned}
\max_{\substack{\mathbf{W}\in\mathcal{W}}}\quad & D(\mathbf{p}_2\mathbf{W} \| \mathbf{p}_1\mathbf{W})\\		
\mathrm{s.t.} \quad & I(\mathbf{p}_k,\mathbf{W})\leq \epsilon_k \quad k=1,2\\
\end{aligned}
\end{equation}
where $\mathcal{W}$ is the set of all $M\times N$ row stochastic matrices, and $\epsilon_k \in [0, H(\mathbf{p}_k)]$,   are the permissible upper bounds of the privacy leakages for the two distributions $\bp_1$ and $\bp_2$, respectively.
Using the approximations in Proposition~\ref{Proposition:approximation_MIandRE}, the PUT for the E-IT approximation problem in the high privacy regime with $0\leq\epsilon_k \ll \min(H(\mathbf{p}_1),H(\mathbf{p}_2))$, for all $k \in \{1,2\}$, is 
\begin{subequations}\label{eq:BHT_approximation}
		\begin{align}
		\label{eq:BHT_approximation_objective}
		\max_{\substack{\mathbf{A}}}\quad & \frac{1}{2}(\mathbf{p}_2-\mathbf{p}_1)\mathbf{A}\mathbf{A}^T(\mathbf{p}_2-\mathbf{p}_1)^T\\
		\label{eq:BHTapprox_privacyconstraint}
		\mathrm{s.t.} \quad &\frac{1}{2}\sum_{i=1}^{M}p_{ki}\|\mathbf{A}_i\|^2\leq \epsilon_k\quad k=1,2\\
		\label{eq:BHTapprox_Aconstraint}
		& \mathbf{A}(\sqrt{\mathbf{w}_0})^T=\mathbf{0}.
		\end{align}
\end{subequations}
where $\mathbf{A}_{i}$ is the $i$-th row of the $M\times N$ matrix $\mathbf{A}$, $\bw_0$ is an interior point of the $N$-dimensional probability simplex, and $\sqrt{\mathbf{w}_0}$ is a row vector with the $i^{\text{th}}$ entry being the squared root of the $i^{\text{th}}$ entry of $\bw_0$, i.e., $\sqrt{w_{0,i}}$. \\
\begin{remark}
	The functions in \eqref{eq:BHT_approximation_objective} and \eqref{eq:BHTapprox_privacyconstraint} are the E-IT approximations as presented in Proposition \ref{Proposition:approximation_MIandRE}, and the constraint \eqref{eq:BHTapprox_Aconstraint} results from the requirement that $\bW$ is row stochastic. This constraint is the only one in \eqref{eq:BHT_approximation} that explicitly involves the size of the output alphabet, i.e., the length of~$\bw_0$.
\end{remark}   

\begin{theorem}\label{Theorem:BHT_approximation_simplification}
	The optimization problem in \eqref{eq:BHT_approximation} reduces to one with a vector variable $\mathbf{a} \in \mathbb{R}^{M }$ as
	\begin{equation}\label{eq:BHT_approx_2}
	\begin{aligned}
	\max_{\substack{\mathbf{a}}}\quad & 
	\frac{1}{2}\|\mathbf{a}(\mathbf{p}_2-\mathbf{p}_1)^T\|^2\\
	\mathrm{s.t.} \quad &\frac{1}{2}\mathbf{a}[\mathbf{p}_k]\mathbf{a}^T\leq \epsilon_k\quad k=1,2\\
	\end{aligned}
	\end{equation}
	where the absolute value of the $i^{\mathrm{th}}$ entry $a_i$ of $\mathbf{a}$, for all $i \in \{1,..,M\}$, is the Euclidean norm of the $i^{\mathrm{th}}$ row $\mathbf{A}_i$ of $\mathbf{A}$. The $M\times N$ matrix $\mathbf{A}^*$ optimizing \eqref{eq:BHT_approximation} is obtained from the optimal solution $\mathbf{a}^*$ of \eqref{eq:BHT_approx_2} as 
	a rank-1 matrix whose $i^{\mathrm{th}}$ row, for all $i$, is given by $a_i^*\mathbf{v}$ where $a_i^*$ is the $i^{\mathrm{th}}$ entry of  $\mathbf{a}^*$, and $\mathbf{v}$ is a unit-norm $N$-dimensional vector that is orthogonal to the non-zero-entry vector $\sqrt{\mathbf{w}_0}$, such that
	\begin{align}
	\label{eq:lemma1_Aopt}
	&\mathbf{A}^*=(\mathbf{a}^*)^T\mathbf{v}\\
	\label{eq:lemma1_v}
	&\mathbf{v}(\sqrt{\mathbf{w}_0})^T=0\\
	\label{eq:lemma1_vnorm}
	&\|\mathbf{v}\|=1.
	\end{align}    
	Finally, it suffices to restrict the output to a binary alphabet, i.e., $N=2$. 
\end{theorem}
The proof of Theorem \ref{Theorem:BHT_approximation_simplification} is in Appendix \ref{proof:BHT_approximation_simplification}. We briefly summarize the approach. The simplification of \eqref{eq:BHT_approximation} to a vector optimization in \eqref{eq:BHT_approx_2} results from the observation that the privacy constraint \eqref{eq:BHTapprox_privacyconstraint} only restricts the row-norms of the matrix variable $\bA$, whereas $\bA$ affects the objective~\eqref{eq:BHT_approximation_objective} through all inner products of rows in $\bA$. By exploiting this special structure, we simplify \eqref{eq:BHT_approximation} to a quadratically constrained quadratic program (QCQP) with a vector variable $\ba$ which governs the Euclidean norms of the rows in $\bA$. The optimal $\bA^*$ is then given by \eqref{eq:lemma1_Aopt} such that the row vector $\bv$ is chosen to satisfy \eqref{eq:BHTapprox_Aconstraint}. Since \eqref{eq:lemma1_v} can be satisfied by a 2-dimensional $\bv$, we conclude  that a binary output alphabet suffices.

Note that the objective function and constraints of the QCQP in \eqref{eq:BHT_approx_2} are ``even'' functions,  i.e., if $\ba$ is feasible, so is its  negation   $-\ba$ and both of them yield the same objective value. Using this observation,  we derive a convex program by removing the square in the objective function. The following theorem provides a closed-form privacy mechanism for the PUT \eqref{eq:BHT_approximation} in high privacy regime by using the Karush-Kuhn-Tucker (KKT) conditions for convex programs.
\begin{theorem}\label{Theorem:BHT_optimalsol_HPapprox}
	An optimal privacy mechanism $\mathbf{W}'$ for the approximation problem in \eqref{eq:BHT_approximation} is
	\begin{align}
	\label{eq:W_generation}
	\mathbf{W}'=\mathbf{W}_0+(\mathbf{a}^*)^T\mathbf{v}\cdot\big[\sqrt{\mathbf{w}_0}\big]
	\end{align}
	where $\mathbf{W}_0$ is given by Proposition \ref{Proposition:approximation_MIandRE}, $\mathbf{v}$ is chosen to satisfy~\eqref{eq:lemma1_v} and \eqref{eq:lemma1_vnorm}, and for $\lambda_{\mathrm{p}}=\|\mathbf{p}_2-\mathbf{p}_1\|^2$ and $\mathbf{v}_{\mathrm{p}}=\frac{\mathbf{p}_2-\mathbf{p}_1}{\|\mathbf{p}_2-\mathbf{p}_1\|}$ being the eigenvalue and eigenvector  of $(\mathbf{p}_2-\mathbf{p}_1)^T(\mathbf{p}_2-\mathbf{p}_1)$, the optimal solution of \eqref{eq:BHT_approx_2}, namely $\mathbf{a}^*$, is given as:
	\begin{enumerate}
		\item if only the first constraint in \eqref{eq:BHT_approx_2} is active,
		\begin{align}
		\label{eq:const1_active}
		\frac{\mathbf{v}_{\mathrm{p}}\Big[\frac{\mathbf{p}_2}{(\mathbf{p}_1)^2}\Big](\mathbf{v}_{\mathrm{p}})^T}{\mathbf{v}_{\mathrm{p}}\big[(\mathbf{p}_1)^{-1}\big](\mathbf{v}_{\mathrm{p}})^T} < \frac{\epsilon_2}{\epsilon_1},
		\end{align}
		and the optimal solution $\mathbf{a}^*$ is
		\begin{align}
		\label{eq:optalpha_const1}
		\mathbf{a}^*=\pm\sqrt{\frac{2\epsilon_1}{\mathbf{v}_{\mathrm{p}}\big[(\mathbf{p}_1)^{-1}\big](\mathbf{v}_{\mathrm{p}})^T}}\mathbf{v}_{\mathrm{p}}\big[(\mathbf{p}_1)^{-1}\big];
		\end{align}
		\item if only the second constraint in \eqref{eq:BHT_approx_2} is active,
		\begin{align}
		\label{eq:const2_active}
		\frac{\mathbf{v}_{\mathrm{p}}\Big[\frac{\mathbf{p}_1}{(\mathbf{p}_2)^2}\Big](\mathbf{v}_{\mathrm{p}})^T}{\mathbf{v}_{\mathrm{p}}\big[(\mathbf{p}_2)^{-1}\big](\mathbf{v}_{\mathrm{p}})^T} < \frac{\epsilon_1}{\epsilon_2},
		\end{align}
		and the optimal solution $\mathbf{a}^*$ is
		\begin{align}
		\label{eq:optalpha_const2}
		\mathbf{a}^*=\pm\sqrt{\frac{2\epsilon_2}{\mathbf{v}_{\mathrm{p}}\big[(\mathbf{p}_2)^{-1}\big](\mathbf{v}_{\mathrm{p}})^T}}\mathbf{v}_{\mathrm{p}}\big[(\mathbf{p}_2)^{-1}\big];
		\end{align}	
		\item when both constraints in \eqref{eq:BHT_approx_2} are active, the optimal solution $\mathbf{a}^*$ is
		\begin{align}
		\label{eq:optimalalpha}
		\mathbf{a}^*=\pm\frac{\lambda_{\mathrm{p}}}{2}\mathbf{v}_{\mathrm{p}}\Big(\eta_1^*\big[\mathbf{p}_1\big]+\eta_2^*\big[\mathbf{p}_2\big]\Big)^{-1}
		\end{align}
		where $\eta_1^*>0$ and $\eta_2^*>0$ satisfy
		\begin{align}
		\label{eq:eta1} \mathbf{v}_{\mathrm{p}}\big[\mathbf{p}_1\big]\Big(\eta_1^*\big[\mathbf{p}_1\big]+\eta_2^*\big[\mathbf{p}_2\big]\Big)^{-2}(\mathbf{v}_{\mathrm{p}})^T&=\frac{8\epsilon_1}{(\lambda_{\mathrm{p}})^2}\\
		\label{eq:eta2}
		\mathbf{v}_{\mathrm{p}}\big[\mathbf{p}_2\big]\Big(\eta_1^*\big[\mathbf{p}_1\big]+\eta_2^*\big[\mathbf{p}_2\big]\Big)^{-2}(\mathbf{v}_{\mathrm{p}})^T&=\frac{8\epsilon_2}{(\lambda_{\mathrm{p}})^2}.
		\end{align}
	\end{enumerate}	
\end{theorem}
The proof of Theorem \ref{Theorem:BHT_optimalsol_HPapprox} involves proving two lemmas and is developed in Appendix \ref{proof:Theorem_BHT_optimalsol_HPapprox}.
\begin{remark}	
	\begin{figure}
		\centering
		\includegraphics[width = \columnwidth]{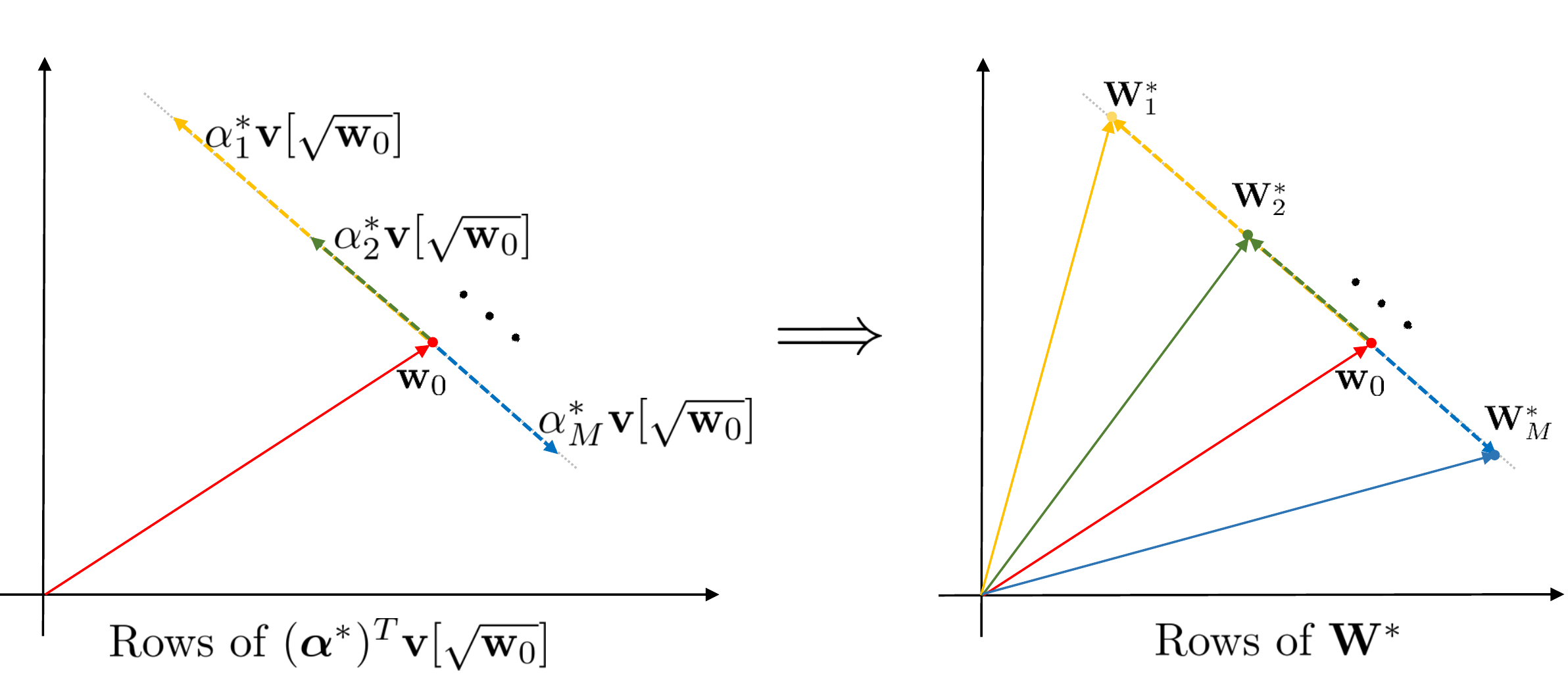}
		\caption{Illustration of  \eqref{eq:W_generation}. The ``nominal'' row vector $\mathbf{w}_0$ is perturbed in different directions depending on $\mathbf{a}^*$ and $\mathbf{v}$.}
		\label{fig:interp}
	\end{figure}
The optimal mechanism $\mathbf{W}'=\mathbf{W}_0+(\mathbf{a}^*)^T\mathbf{v}\big[\sqrt{\mathbf{w}_0}\big]$  captures the fact that a statistical privacy metric, such as MI, takes into consideration the source distribution in designing the perturbation mechanism $\boldsymbol{\Theta}^*$. In fact, the solutions for $\mathbf{a}^*$ in \eqref{eq:optalpha_const1}, \eqref{eq:optalpha_const2} and \eqref{eq:optimalalpha} quantify this through the term $\big(\eta_1^*[\mathbf{p}_1]+\eta_2^*[\mathbf{p}_2]\big)^{-1}$. The vector $\mathbf{v}_{\mathrm{p}}$ indicates the direction along which the objective function, i.e., the relative entropy, grows the fastest. In Fig.~\ref{fig:interp} we illustrate the results of Theorem \ref{Theorem:BHT_optimalsol_HPapprox}. Thus,
for  a uniformly distributed source, all entries of $\mathbf{v}_{\mathrm{p}}$ have the same scaling such that $\mathbf{a}^*$ is in the direction of $\mathbf{v}_{\mathrm{p}}$. However, for a non-uniform source, the samples with low probabilities affect the direction of $\mathbf{v}_{\mathrm{p}}\big(\eta_1^*[\mathbf{p}_1]+\eta_2^*[\mathbf{p}_2]\big)^{-1}$ the most. This is a   consequence of the statistical leakage metric  (the MI) which causes the optimal mechanism to minimize information leakage by perturbing the low probability (more informative) symbols proportionately more relative to the higher probability ones. 
\end{remark}

\subsection{Binary Hypothesis Testing (R\'enyi Divergence Setting)}\label{subsection:BHT_renyi}
We now consider the scenario in which both the false alarm and missed detection probabilities for $H_1$ are exponentially decreasing. For this case, the trade-off between the two error probabilities is captured by the R\'enyi divergence as shown in \cite{ EEHT_Tuncel, HypothesisTesting_Origin_Blahut}. We use this as our utility metric and briefly review the results in \cite{ EEHT_Tuncel, HypothesisTesting_Origin_Blahut} as a starting point.

Assume that the false alarm probability decays as $\exp(-nE_{2,1})$, for some exponent $E_{2,1}>0$. Then, the largest error exponent of the missed detection $E_{1,2}$ for a fixed $E_{2,1}$ is a function of $E_{2,1}$ given by \cite{HypothesisTesting_Origin_Blahut}   
\begin{align}\label{eq:BHT_NZEE1}
E_{1,2}(E_{2,1}) \triangleq \min_{\bp: D(\bp \| \bp_2) \le E_{2,1}}D( \bp\|\bp_1).
\end{align}
Since \eqref{eq:BHT_NZEE1} is a convex program, it
can be equivalently characterized by the Lagrangian minimization
\begin{align}
L(\beta)\triangleq\min_{\substack{\bp}}\{ D(\bp\|\bp_1)+ \beta D(\bp\|\bp_2)\} \label{eqn:Lb},
\end{align}
leading to the dual problem ~\cite{HypothesisTesting_Origin_Blahut}  
\begin{equation}
\label{eq:Renyi_dualmaximization}
E_{1,2}(E_{2,1}) = \sup_{\beta\geq 0} \left\{ L(\beta) - \beta E_{2,1} \right\}.
\end{equation}
The $\bp^*$ optimizing \eqref{eqn:Lb} can be computed using the KKT conditions of \eqref{eq:BHT_NZEE1} (cf. \cite[(15)]{EEHT_Tuncel}) to further obtain
\begin{align}\label{eq:lagrangian}
L(\beta)=-(1+\beta)\log\Big(\sum_{x}p_1(x)^{\frac{1}{1+\beta}}p_2(x)^{\frac{\beta}{1+\beta}}\Big).
\end{align}
For $\alpha\triangleq\frac{\beta}{1+\beta}\in (0,1)$, \eqref{eq:lagrangian} simplifies as \cite{EEHT_Tuncel}
\begin{eqnarray}
\frac{1}{\alpha\!-\!  1}\log\Big(\sum_{x}p_2(x)^{\alpha}p_1(x)^{1-\alpha}\Big)\!=\!\Dalpha(\bp_2\|\bp_1)
\end{eqnarray}
where $\Dalpha(\bp_2\|\bp_1)$ is the order-$\alpha$  R\'{e}nyi divergence. From \eqref{eqn:Lb} and \eqref{eq:Renyi_dualmaximization}, we see that $\Dalpha(\bp_2\|\bp_1)$ is the weighted sum of the two error exponents, i.e., $E_{1,2}(E_{2,1})+\frac{\alpha}{1-\alpha} E_{2,1}$, and as such a good candidate for a utility metric in this setting. For this metric, one can write the PUT problem as
\begin{equation}\label{eq:BHT_nonzeroEE}
\begin{aligned}
\max_{\substack{\mathbf{W}\in\mathcal{W}}}\quad & \Dalpha(\bp_2\mathbf{W}\|\bp_1\mathbf{W})\\		
\mathrm{s.t.} \quad & I(\mathbf{p}_k,\mathbf{W})\leq \epsilon_k, \quad k=1,2.\\
\end{aligned}
\end{equation}
Analogous to the PUT in \eqref{eq:OrigProHP} with relative entropy as the utility metric, the optimization in \eqref{eq:BHT_nonzeroEE} is non-convex and NP-hard. Thus, we focus on the high privacy regime and approximate the order-$\alpha$ R\'{e}nyi divergence in that regime. To this end, we use the following lemma to explicitly present the relationship of the order-$\alpha$ R\'{e}nyi divergence $D_{\alpha}(\bp_2\|\bp_1)$ and the relative entropy $D(\bp_2\|\bp_1)$ when $\bp_2$ and $\bp_1$ are ``close''.

\begin{lemma}\label{lemma:Approximation_Renyidiv}
For $\alpha\in(0,1)$, the following continuity statement holds: If \footnote{We say that a vector $\bp_2$ converges  to another vector $\bp_1$, denoted as $\bp_2\to\bp_1$, if $\|\bp_2-\bp_1\|_\infty\to 0$. } $\bp_2\to\bp_1$, then
\begin{equation}
\frac{(1-\alpha) D(\bp_2\|\bp_1)}{  {2^{(1-\alpha ) D_{\alpha}(\bp_2\|\bp_1) }-1}   }\to \frac{\log e}{\alpha}.  \label{eq:approx_RenyiDiverKLDiver}
\end{equation}
\end{lemma}
The proof is detailed in Appendix \ref{proof:lemma_Approximation_Renyidiv}.

According to Proposition \ref{Proposition:approximation_MIandRE}, any privacy mechanism $\bW$ in the high privacy regime is a perturbation of a perfect privacy mechanism $\bW_0$. When $\epsilon_k$ in \eqref{eq:BHT_nonzeroEE} is close to zero, $\rho$ is also close to $0$ and both output distributions $\bp_1\bW$ and $\bp_2\bW$ approach $\bw_0$. We now use Lemma \ref{lemma:Approximation_Renyidiv} in the following corollary to show that the ratio of $(1-\alpha) D(\bp_2\mathbf{W}\|\bp_1\mathbf{W})$ and $2^{(1-\alpha ) D_{\alpha}(\bp_2\mathbf{W}\|\bp_1\mathbf{W}) }-1$ converges to the constant $\alpha^{-1} \log e$.

\begin{corollary}\label{corollary:Renyidiv_KL_BHTNZEE}
	Let  $\alpha\in(0,1)$.  In \eqref{eq:BHT_nonzeroEE},  if $\epsilon_1,\epsilon_2\to 0$, $\bW$ converges to a perfect privacy mechanism $\bW_0$ (cf.\  Lemma \ref{lemma:perfect_privacy}). Consequently, $\bp_2\mathbf{W}\to\bp_1\mathbf{W}$ and the following convergence statement also holds. 
	\begin{align}\label{eq:Renyidiv_KL_BHTNZEE}
	\frac{(1-\alpha) D(\bp_2\mathbf{W}\|\bp_1\mathbf{W})}{  {2^{(1-\alpha ) D_{\alpha}(\bp_2\mathbf{W}\|\bp_1\mathbf{W}) }-1}   }\to \frac{\log e}{\alpha}.
	\end{align}
\end{corollary}
From \eqref{eq:Renyidiv_KL_BHTNZEE}, we observe that as $\epsilon_1,\epsilon_2\to 0$, 
$D_\alpha(\bp_2W\|\bp_1W)$ is monotonically increasing in $D(\bp_2W\|\bp_1W)$. 
Thus, in this high privacy regime,  the optimizer of $D_{\alpha}(\bp_2W\|\bp_1W)$ is the same as $D(\bp_2W\|\bp_1W)$. As a result, in the high privacy regime we revert to the relative entropy setting, for which we provide a closed-form solution in Theorem \ref{Theorem:BHT_optimalsol_HPapprox}.

\section{$m$-ary Hypothesis Testing in the High Privacy Regime}\label{section:MHT}
We now consider the $m$-ary hypothesis testing problem with $m$ distinct hypotheses $H_k$, $k\in\{1,\ldots,m\}$, each corresponding to a distribution $\bp_k$. This in turn results in $m(m-1)$ error probabilities of incorrectly inferring hypothesis $H_i$ as hypothesis $H_j$. As stated in Section~\ref{section:Problem_Formulation}, to simplify our analysis, we consider a scenario somewhat analogous to the ``red alert''~\cite{RedAlert_Nazer} problem in ``unequal error protection'' \cite{UnequalErrorProtection_Borade}, i.e., there is one distinct hypothesis $H_1$, the inference of which is more crucial than that of others (e.g., presence of cancer). We focus on maximizing the minimum of the $m-1$ error exponents corresponding to the $m-1$ ways of incorrectly deciding $H_1$ as $H_j , j \ne 1$. 

For this problem of unequal $m$-ary hypothesis testing, we introduce the PUT in \eqref{eq:MHT_original}. We can further simplify the trade-off in the high privacy regime using Proposition \ref{Proposition:approximation_MIandRE} to obtain the following PUT:
\begin{subequations}\label{eq:MHT_approximation}
	\begin{align}
	\label{eq:MHT_approximation_obj}
		\max_{\substack{\bA}} \quad & \min  \Big\{\frac{1}{2}\|(\bp_k-\bp_1)\bA\|^2, k=2,\ldots,m\Big\}\\
		\label{eq:MHT_approximation_privacy}
		\mathrm{s.t.} \quad &\frac{1}{2}\sum_{i=1}^{M}p_{ki}\|\bA_i\|^2\leq \epsilon_k \quad k=1,2,\ldots,m\\
		\label{eq:MHT_approximation_A}
		& \bA(\sqrt{\bw_0})^T=\mathbf{0}.
	\end{align}
\end{subequations}
Recall that $\bA\in\mathbb{R}^{M\times N}$ is a perturbation matrix such that the privacy mechanism $\bW$ is related to $\bW_0$ as $\bW=\bW_0+\bA[\sqrt{\bw_0}]$, and $\bA_i$ is the $i^{\text{th}}$ row of $\bA$.\\

For ease of analysis, we start from a simplified version of \eqref{eq:MHT_approximation} without the constraint \eqref{eq:MHT_approximation_A}, which can be transformed to a semi-definite program (SDP) as summarized in the following lemma. Based on an optimal solution of the SDP, a scheme is proposed for constructing an optimal solution $\bA^*$ of \eqref{eq:MHT_approximation} satisfying \eqref{eq:MHT_approximation_A}.

\begin{lemma}\label{Lemma:MHT_approximation_simplerSDP}
The optimization in \eqref{eq:MHT_approximation_obj} with constraint \eqref{eq:MHT_approximation_privacy} is equivalent to an SDP  with ($M\times M$ matrix) variable $\bB=\bA\bA^T$ given as
	\begin{align}\label{eq:MHT_approximation_simplerSDP}
			\max_{\substack{\bB,t}} \quad &  t \nonumber\\
			\mathrm{s.t.} \quad &\frac{1}{2}\Tr\big((\bp_k-\bp_1)^T(\bp_k-\bp_1)\bB\big) \geq t\quad k=2,\ldots,m\nonumber\\
			&\frac{1}{2}\Tr\big([\bp_k]\bB\big) \leq \epsilon_k  \quad k=1,2,\ldots,m\nonumber\\
			&\bB\succeq 0
	\end{align}
	where $[\bp_k]$ is a diagonal matrix with $i^{\text{th}}$ diagonal entry equal to $p_{ki}$, and $\Tr[[\bp_k]\bB]$ is the trace  of the matrix $[\bp_k]\bB$.
\end{lemma}
Lemma \ref{Lemma:MHT_approximation_simplerSDP} stems from the observation that both the objective function \eqref{eq:MHT_approximation_obj} and constraints \eqref{eq:MHT_approximation_privacy} are linear functions of the entries of the positive semidefinite matrix $\bA\bA^T$. The proof for Lemma \ref{Lemma:MHT_approximation_simplerSDP} is provided in Appendix \ref{proof:lemma_MHT_approximation_simplerSDP}. The following theorem shows that the solution of the SDP in \eqref{eq:MHT_approximation_simplerSDP} yields an optimal privacy mechanism for the approximated PUT in \eqref{eq:MHT_approximation}.
\begin{theorem}\label{Theorem: MHT_EITappr_optimalW}
	An optimal \text{privacy mechanism} $\bW'$ for the optimization problem in \eqref{eq:MHT_approximation} is 
	\begin{align}
	\label{eq:MHTrecovered_W}
	\bW'=\bW_0+\bA^*[\sqrt{\bw_0}]
	\end{align}
	where $\bW_0$ is the perfect privacy mechanism with   rows $\bw_0$, $\bA^*$ is an optimal solution of~\eqref{eq:MHT_approximation} obtained from an optimal solution $\bB^*=\bU^*[\blambda](\bU^*)^T$ with $l\triangleq \mathrm{rank}(\bB^*)\leq M$ of the SDP in~\eqref{eq:MHT_approximation_simplerSDP}. It suffices to restrict the output alphabet size $N$ to $l+1$, such that
	\begin{align}
		\label{eq:MHTrecovered_A}
		\bA^*=\bU^*\bSigma^*\bV^T,
	\end{align}
		where $\bSigma^*$ is an $M\times (l+1)$ rectangular diagonal matrix whose diagonal entries are the square roots of the $l$ non-zero eigenvalues $\blambda$ of $\bB^*$, $\bU^*$ is a unitary matrix consisting of the eigenvectors of $\bB^*$, and $\bV$ is an $(l+1)\times (l+1)$  unitary matrix whose the first $l$ columns are orthogonal to $\sqrt{\bw_0}$.
\end{theorem} 
\begin{proof}
	Let $\bB^*$ be the optimal solution of the SDP in~\eqref{eq:MHT_approximation_simplerSDP} and $l\triangleq \mathrm{rank}(\bB^*)$. We decompose $\bB^*$ via an  eigenvalue decomposition as follows:
	\begin{align}
	\bB^*=\bU^*[\blambda](\bU^*)^T.
	\end{align}  
	Here, $[\blambda]$ is an $l\times l$ diagonal matrix consisting of entries in the eigenvalue vector $\blambda$ and the columns of the $M\times l$ matrix $\bU^*$ are the $l$ corresponding eigenvectors. Construct an $l\times (l+1)$ rectangular diagonal matrix $\bSigma^*$ by adding one all-zero column to $[\sqrt{\blambda}]$. Let $N=l+1$. By choosing a $(l+1)\times (l+1)$ unitary matrix $\bV$, whose last column parallel to the $(l+1)$-dimensional row vector $\sqrt{\bw_0}$, we design a matrix $\bA^*$ as $\bU^*\bSigma^*\bV^T$ such that $\bA^*(\bA^*)^T=\bU^*\bSigma^*(\bSigma^*)^T(\bU^*)^T=\bB^*$. From Lemma~\ref{Lemma:MHT_approximation_simplerSDP}, the SDP in \eqref{eq:MHT_approximation_simplerSDP} is equivalent to~ the simplified \eqref{eq:MHT_approximation} without \eqref{eq:MHT_approximation_A}. Therefore, $\bA^*$ optimizes~the simplified \eqref{eq:MHT_approximation}. 
	In addition, 
	\begin{align}
	& \bA^*(\sqrt{\bw_0})^T=\bU^*\bSigma^*\bV^T(\sqrt{\bw_0})^T \\
	\label{eq:MHT_VerifyconstructV}
	&=\bU^*\begin{bmatrix}
	\sqrt{\lambda_1}&& &0 \\
	&\ddots&  &\vdots \\
	&& \sqrt{\lambda_l}&0 
	\end{bmatrix}_{l\times (l+1)}\begin{bmatrix}
	0 \\
	\vdots\\
	0 \\
	\|\sqrt{\bw_0}\| \\
	\end{bmatrix}_{(l+1)}\\
	&=\bU^*\mathbf{0}=\mathbf{0}
	\end{align}
	where \eqref{eq:MHT_VerifyconstructV} follows from the fact that the last column of the $(l+1)\times (l+1)$ unitary matrix $\bV$ is parallel to $\sqrt{\bw_0}$, such that the first $l$ columns of $\bV$ are orthogonal to $\sqrt{\bw_0}$, and the inner product of its last column and $\sqrt{\bw_0}$ is the Euclidean norm of $\sqrt{\bw_0}$. 
	Therefore, the $\bA^*$ constructed above is feasible and attains the optimal value of \eqref{eq:MHT_approximation}.
\end{proof}
\begin{remark}
Note that the size of output alphabet is at most $M+1$. For the special case of binary hypothesis testing, we have shown in Theorem \ref{Theorem:BHT_approximation_simplification} that the rank of $B^*$ is $1$ and therefore, $N=2$.
\end{remark}
\begin{remark}
In the absence of any constraints in \eqref{eq:MHT_approximation}, analogous to the binary hypothesis test, one would choose $\min\{m-1,M-1\}$ columns of $\bU$ to span the space contained by the vectors $\bp_k-\bp_1$ for all $k\neq 1$. However, the constraints in \eqref{eq:MHT_approximation} depend explicitly on the vectors $\bp_k$, and in fact, in \eqref{eq:MHT_approximation_simplerSDP} at least one constraint will be tight at the optimal solution $\bB^*$. Thus, analogous to the binary hypothesis result, we expect the optimal mechanism to depend inversely on one or more $\bp_k$. We show that this is indeed the case for binary sources in the following subsection.
\end{remark}
\vspace{-0.5cm}

\subsection{$m$-ary Hypotheses Testing with Binary Sources}\label{subsection_MHT_BinarySources}
If all the $m$ distributions $\bp_k$ are Bernoulli, the $m-1$ difference vectors $\bp_k-\bp_1$ in  \eqref{eq:MHT_approximation} are {\em collinear}. 
Thus, the minimizing element in the objective is the one in which $\bp_k-\bp_1$ has the minimal Euclidean norm. Without loss of generality, assume $\|\bp_2-\bp_1\|=\min\{\|\bp_k-\bp_1\|,k=2,\ldots,m\}$. Therefore, $\|(\bp_2-\bp_1)\bA\|^2=\min\{\|(\bp_k-\bp_1)\bA\|^2,k=2,\ldots,m\}$. In this case, the E-IT approximation in~\eqref{eq:MHT_approximation} reduces to
\begin{subequations}\label{eq:BinarySource_approximation}
	\begin{align}
	\label{eq:BinarySource_approximation_obj}
		\max_{\substack{\bA}} \quad & \frac{1}{2}\|(\bp_2-\bp_1)\bA\|^2\\
		\label{eq:BinarySource_approximation_PrivacyCons}
		\mathrm{s.t.} \quad &\frac{1}{2}\sum_{i=1}^{2}p_{ki}\|\bA_i\|^2\leq \epsilon_k \quad \ k=1,\ldots,m\\
		\label{eq:BinarySource_approximation_A}
		& \bA(\sqrt{\bw_0})^T=\mathbf{0}.
	\end{align}
\end{subequations}

We notice that \eqref{eq:BinarySource_approximation} has the same form as \eqref{eq:BHT_approximation} (the E-IT approximation for binary hypothesis testing for the relative entropy setting), where the number of constraints in \eqref{eq:BinarySource_approximation_PrivacyCons} is $m\geq 2$. Specifically, the objective and constraints have the same structure as in \eqref{eq:BHT_approximation}, and thus, the results in Theorem \ref{Theorem:BHT_approximation_simplification} holds here. Therefore, from Theorem \ref{Theorem:BHT_optimalsol_HPapprox}, the corresponding optimal privacy mechanism can be expressed as \eqref{eq:W_generation} but with
\begin{align}
\label{eq:BinarySource_approximation_opta}
\ba^*=\frac{\bp_2-\bp_1}{2}\left(\left[\sum_{k=1}^{m}\eta^*_k\bp_k\right]\right)^{-1},
\end{align}
where $\eta^*_k\geq 0$, $k=1,\ldots,m$, are the dual variables for the $m$ constraints in \eqref{eq:BinarySource_approximation_PrivacyCons}.

 Note that for those $\eta_k^*$ that are non-zero,  the corresponding constraints in \eqref{eq:BinarySource_approximation_PrivacyCons} are tight, i.e., if $\eta^*_k>0$, $\frac{1}{2}\sum_{i=1}^{2}p_{ki}\|\bA_i\|^2=\epsilon_k$. Let $\mathcal{K}=\{k: \eta_k^*>0,\,k=1,\ldots,m\}$. Thus, the $\ba^*$ in \eqref{eq:BinarySource_approximation_opta} depends inversely on a linear combination of the distributions indexed by $\mathcal{K}$. Consequently, the optimal mechanism for the approximated PUT depends inversely on these distributions.

\section{Numerical Results}\label{section:illustration}
In this section, we numerically evaluate the utilities achieved by optimal privacy mechanisms for E-IT approximations in  two scenarios: $m=2$ (binary) and $m=3$ (ternary) hypothesis testing. Furthermore, for the binary hypothesis testing scenario, we consider both the relative entropy and R\'enyi divergence settings, while for the $m=3$ scenario, we only focus on the relative entropy setting. Our goal is to compare the maximal utility for the E-IT approximation with that achieved for the original PUT. To this end, we start by choosing a privacy leakage level $\epsilon_k=\tilde{\epsilon}\ll \min_{k}H(\bp_k)$, for all $k\in\{1,\ldots,m\}$, for the E-IT approximation.

Recall that for the relative entropy setting, \eqref{eq:MHTrecovered_W} in Theorem~\ref{Theorem: MHT_EITappr_optimalW} provides an optimal privacy mechanism $\mathbf{W}'(\tilde{\epsilon})$ for the E-IT approximation problem in \eqref{eq:MHT_approximation} with leakage bounds $\epsilon_k=\tilde{\epsilon}$ for all $k$. Specifically, for $m=2$, $\mathbf{W}'(\tilde{\epsilon})$ can also be expressed as \eqref{eq:W_generation} in Theorem~\ref{Theorem:BHT_optimalsol_HPapprox}, where $\ba^*$ and $\bv$ are the first columns of $\bU^*\bSigma^*$ and $\bV$ in \eqref{eq:MHTrecovered_W}, respectively. From Corollary~\ref{corollary:Renyidiv_KL_BHTNZEE}, in the high privacy regime $\mathbf{W}'(\tilde{\epsilon})$ in \eqref{eq:W_generation} is also the optimal mechanism (for the approximated PUT) for binary hypothesis testing in R\'enyi divergence setting. 

To evaluate the performance of $\mathbf{W}'$, we compare its utility to that achieved by an optimal mechanism $\bW^*$ of the original PUT problem (e.g., \eqref{eq:MHT_original} for the relative entropy setting or \eqref{eq:BHT_nonzeroEE} for the R\'enyi divergence setting). For a fair comparison of the utilities resulting from the E-IT and original PUTs, we choose the MI leakages to be the same for both cases. Thus, for the relative entropy setting (resp.\ R\'enyi divergence setting), we compare the values of $D(\mathbf{p}_2\mathbf{W}'(\tilde{\epsilon}) \| \mathbf{p}_1\mathbf{W}'(\tilde{\epsilon}))$ (resp.\ $D_{\alpha}(\mathbf{p}_2\mathbf{W}' (\tilde{\epsilon})\| \mathbf{p}_1\mathbf{W}'(\tilde{\epsilon}))$) and $D(\mathbf{p}_2\mathbf{W}^*(\epsilon) \| \mathbf{p}_1\mathbf{W}^*(\epsilon))$ (resp.\ $D_{\alpha}(\mathbf{p}_2\mathbf{W}^*(\epsilon) \| \mathbf{p}_1\mathbf{W}^*(\epsilon))$), where $\epsilon=\max_k I(\bp_k,\mathbf{W}'(\tilde{\epsilon}))$.

For the original PUT problems in \eqref{eq:MHT_original} and \eqref{eq:BHT_nonzeroEE}, the number of independent variables in $\mathbf{W}$ is $M(M-1)$ for $M=N$. Even for $M=3$, finding the optimal privacy mechanism $\mathbf{W}^*(\epsilon)$ using exhaustive search techniques is computationally prohibitive. Therefore, we restrict our numerical analysis to binary sources, i.e., $M=2$; furthermore, for numerical tractability in computing $\bW^*(\epsilon)$, we assume that $M=N=2$, i.e., the output alphabet is binary.  

For the E-IT approximated PUTs, since the choice of $\bw_0$ does not affect the optimal $\bA^*$, we choose $\mathbf{w}_0=(0.5, 0.5)$, for which from \eqref{eq:lemma1_v} and \eqref{eq:lemma1_vnorm}, we have $\mathbf{v}=\pm (\sqrt{0.5} ,-\sqrt{0.5})$. To capture the high privacy regime, we restrict $\tilde{\epsilon}\leq 0.2\min_{k}H(\bp_k)$. For these parameters, the following two subsections illustrate and discuss the regimes in which the E-IT approximation is accurate.

\begin{table*}
	\centering
	\begin{tabular}{|c|c|c|c|}
		\hline
		$\mathbf{p}_1$ $\mathbf{p}_2$ &  close to each other & close to the uniform distribution\\ \hline
		Pair 1: $(0.55,0.45)$ $(0.95,0.05)$   &  No 				 &$\bp_1$ Yes, $\bp_2$ No\\ \hline
		Pair 2: $(0.95,0.05)$ $(0.05,0.95)$   &  No 				 &$\bp_1$ No, $\bp_2$ No\\ \hline
		Pair 3: $(0.50,0.50)$ $(0.45,0.55)$   &  Yes 				 &$\bp_1$ Yes, $\bp_2$ Yes\\ \hline
		Pair 4: $(0.10,0.90)$ $(0.05,0.95)$   &  Yes 				 &$\bp_1$ No, $\bp_2$ No\\ \hline	
	\end{tabular}
	\caption{Distribution pairs for binary hypothesis testing}
	\label{Table:BHT_binarySources}
\end{table*} 	
	
\subsection{Binary Hypothesis Testing} 
We consider four pairs of Bernoulli distributions as shown in Table \ref{Table:BHT_binarySources} for the two source classes (hypotheses) to evaluate the accuracy of optimal mechanisms for the E-IT approximation in the relative entropy and R\'enyi divergence settings. Figures \ref{subfig:fig1_1}-\ref{subfig:fig1_4} illustrate the normalized utilities for Pairs 1-~4 in Table \ref{Table:BHT_binarySources}, respectively, as a function of the normalized MI leakages, i.e., $\epsilon/\min\{H(\mathbf{p}_1),H(\mathbf{p}_2)\}$. In the four figures, the left and  right $y$-axes are for normalized utilities in the relative entropy and R\'enyi divergence settings, respectively. 
Figures~\ref{subfig:fig1_1} and \ref{subfig:fig1_4} show that $\bW'$ and $\bW^*$ have the same utilities in the regions highlighted by the black-dotted ellipses, in which $\epsilon$ is smaller than $0.5\%$ and $0.1\%$ of $\min\{H(\mathbf{p}_1),H(\mathbf{p}_2)\}$, respectively. In contrast, for Figs.~\ref{subfig:fig1_2} and~\ref{subfig:fig1_3}, the utilities of $\bW'$ and $\bW^*$ are almost the same in the entire plotted range.
\begin{figure*}
	\centering
	\begin{subfigure}{0.5\textwidth}
		\centering
		\includegraphics[width=.9\columnwidth]{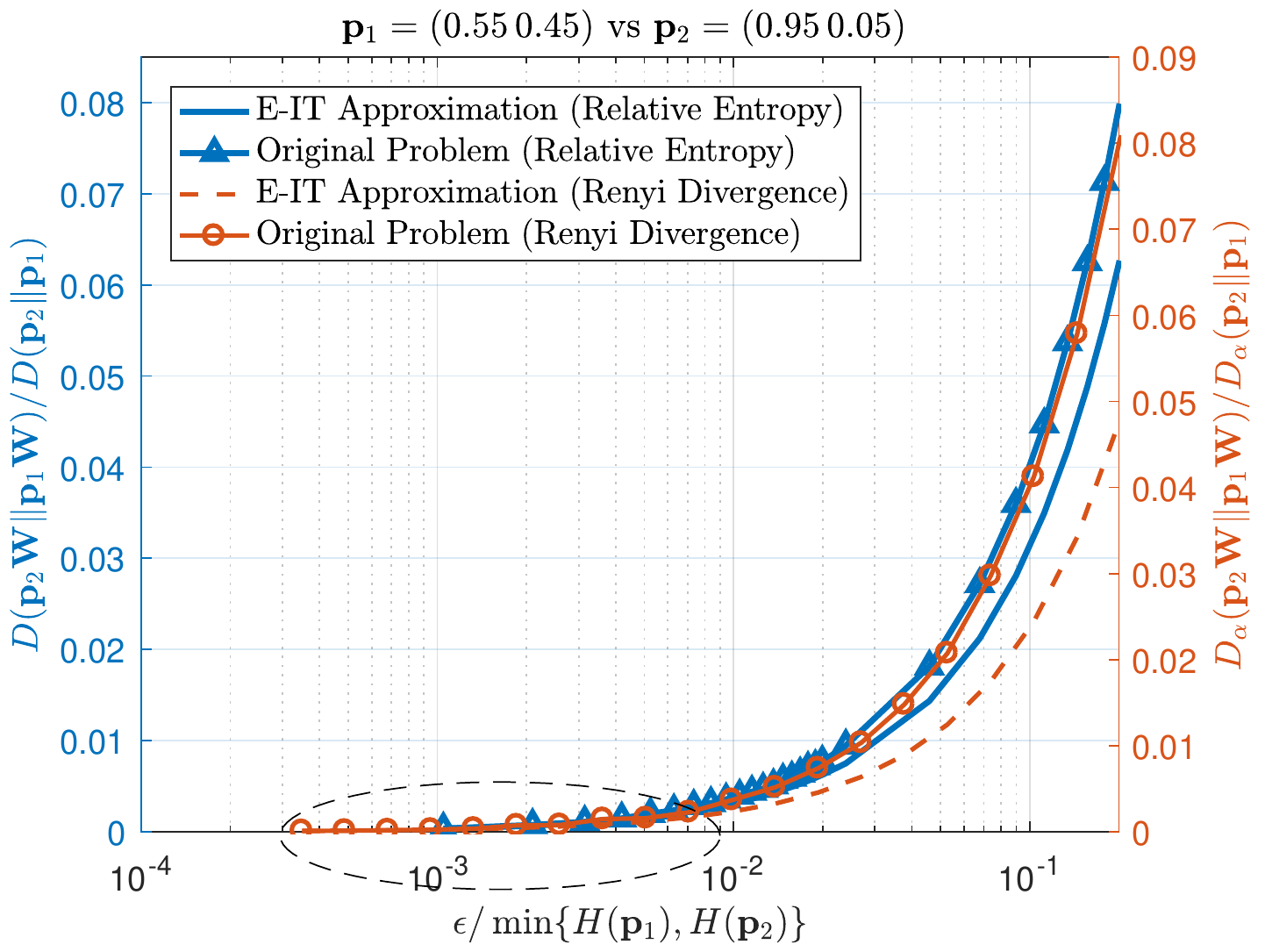}
		\caption{Pair 1}
		\label{subfig:fig1_1}
	\end{subfigure}%
	\begin{subfigure}{0.5\textwidth}
		\centering
		\includegraphics[width=.9\columnwidth]{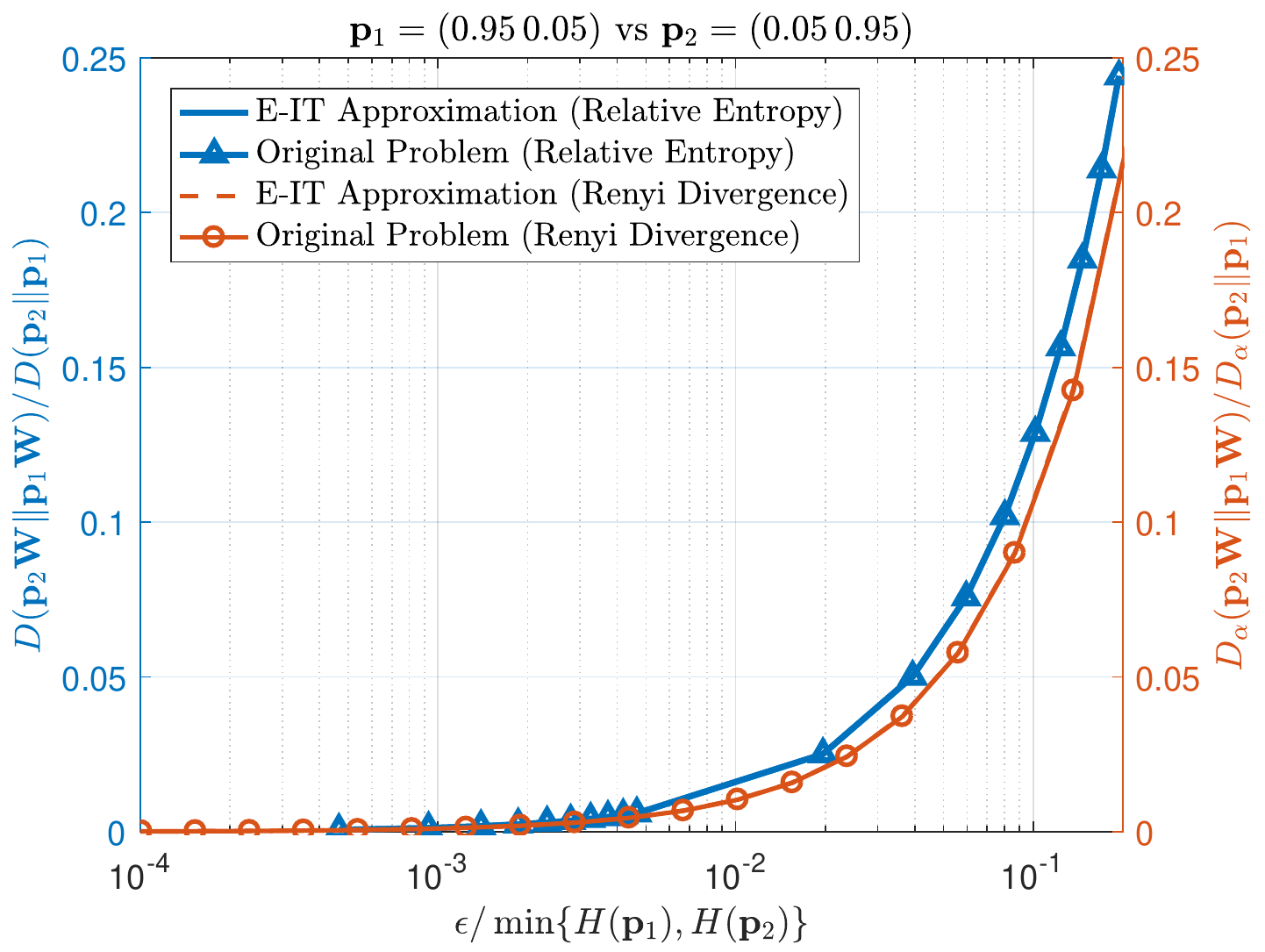}
		\caption{Pair 2}
		\label{subfig:fig1_2}
	\end{subfigure}\\
	\begin{subfigure}{0.5\textwidth}
		\centering
		\includegraphics[width=.9\columnwidth]{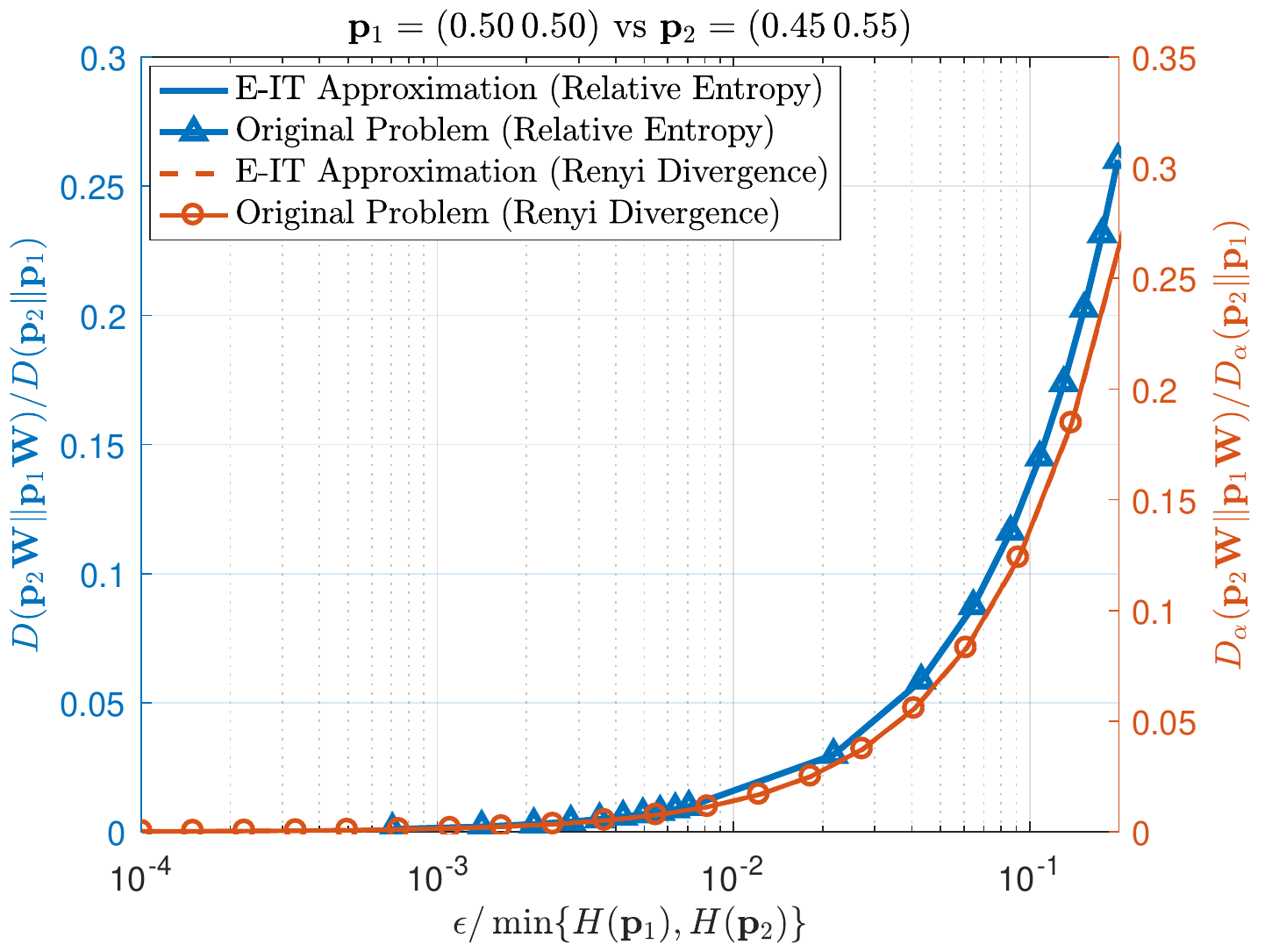}
		\caption{Pair 3}
		\label{subfig:fig1_3}
	\end{subfigure}%
	\begin{subfigure}{0.5\textwidth}
		\centering
		\includegraphics[width=.9\columnwidth]{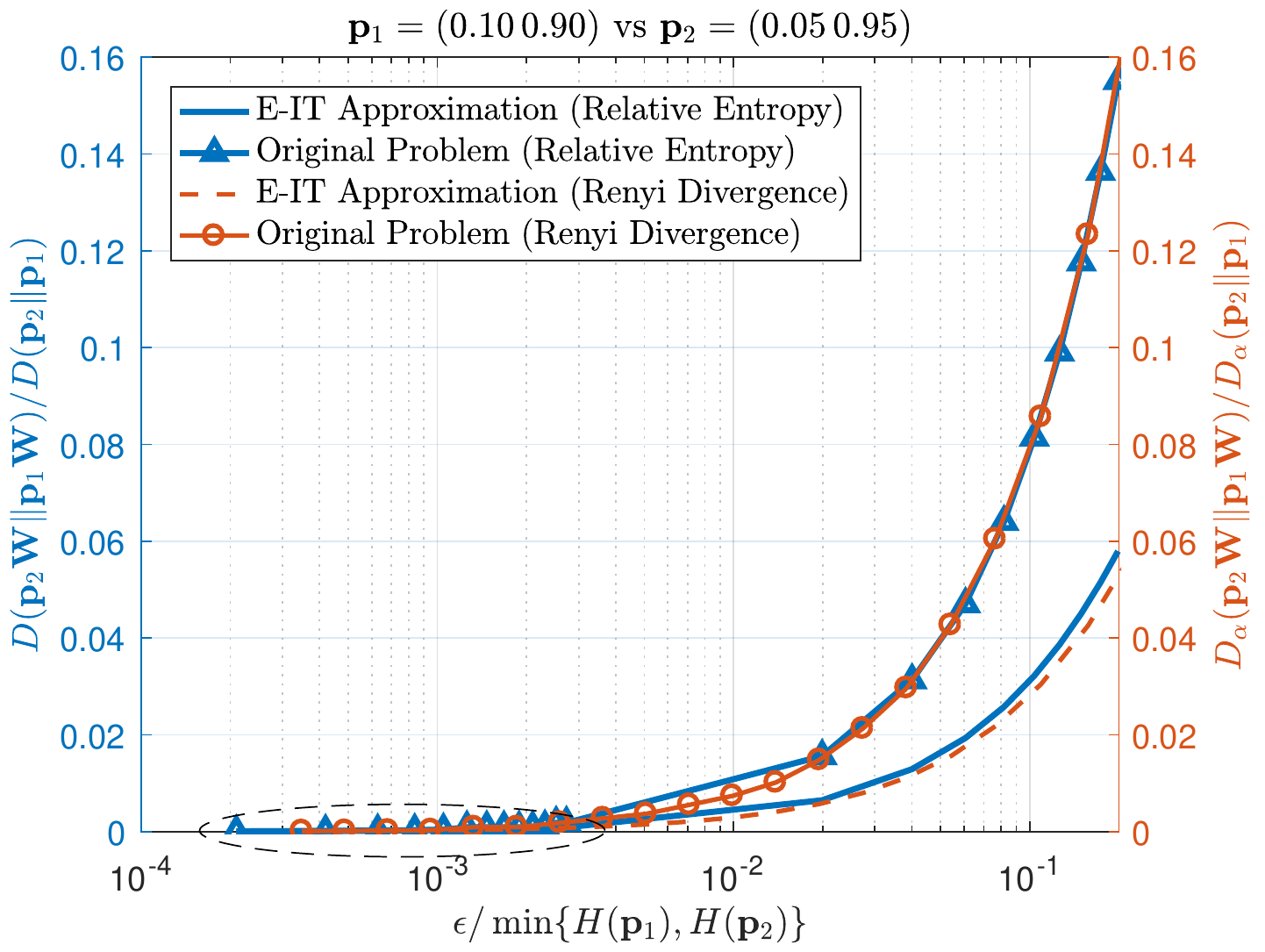}
		\caption{Pair 4}
		\label{subfig:fig1_4}
	\end{subfigure}\\
	\caption{The relative utilities of $\mathbf{W}'$ and $\mathbf{W}^*$ for the four distribution pairs in Table \ref{Table:BHT_binarySources}}
	\label{fig:fig1}
\end{figure*}

From Figs.~\ref{subfig:fig1_1}--\ref{subfig:fig1_4}, we deduce that for any two given distributions, there is high privacy regime in which the performance of the privacy mechanism $\bW'$ for the E-IT approximation is almost optimal; however, the range of the regime is specific to the distribution pairs. In particular, when both distributions are close to the uniform or when both are far apart from the uniform as well as each other, the set of leakage values for which the privacy mechanism $\bW'$ works well is larger. For the former, it can be seen that the E-IT approximations of the relative entropy and the MI are more accurate (cf.~\cite[Footnote~2]{EIT2015}); for the latter, the individual  approximation errors  ``cancel out'' so the overall approximation is accurate.

\subsection{Ternary Hypothesis Testing}
\begin{table}[t]
	\centering
	\begin{tabular}{|c|c|}
		\hline
		$\mathbf{p}_1$, $\mathbf{p}_2$, $\mathbf{p}_3$\\ \hline
		 Triple 1 $(0.50,0.50)$ $(0.45,0.55)$ $(0.55 ,0.45)$ \\ \hline
		 Triple 2 $(0.15,0.85)$ $(0.10,0.90)$ $(0.20 ,0.80)$ \\ \hline
	\end{tabular}
	\caption{Distribution triples for ternary hypothesis testing}
	\label{Table:MHT_3binarySources}
\end{table} 
We numerically evaluate our results for ternary hypothesis testing using three Bernoulli distributions, one for each of the three hypotheses. As shown in Table \ref{Table:MHT_3binarySources}, we consider two such triples. Figures \ref{subfig:MHTbs_1} and \ref{subfig:MHTbs_4} illustrate the normalized utilities for Triples 1 and 2, respectively, as a function of the normalized MI leakages, i.e., $\epsilon/\min_{k}\{H(\mathbf{p}_k),k\in\{1,2,3\}\}$. 

Fig.~\ref{subfig:MHTbs_1} shows that the normalized utilities for $\bW'$ and $\bW^*$ are almost the same in the entire plotted range for Triple 1. As shown in Fig.~\ref{subfig:MHTbs_4}, for Triple 2, the normalized utilities for $\mathbf{W}'$ and $\mathbf{W}^*$ are close only in the region where the MI-based leakages $\epsilon$ are less than $0.2\%$ of  $\min_k H(\mathbf{p}_k$. As for the binary case, here too, our plots show that the leakage range for which the approximation is tight depends on the distributions. For Triple 1, the good performance of the optimal mechanism $\bW'$ exists for a larger set of MI leakages than for Triple 2, because the three distributions in Triple 1 are close to uniform distribution such that the E-IT approximation is accurate.

\begin{figure*}
	\centering
	\begin{subfigure}{0.5\textwidth}
		\centering
		\includegraphics[width=.79\columnwidth]{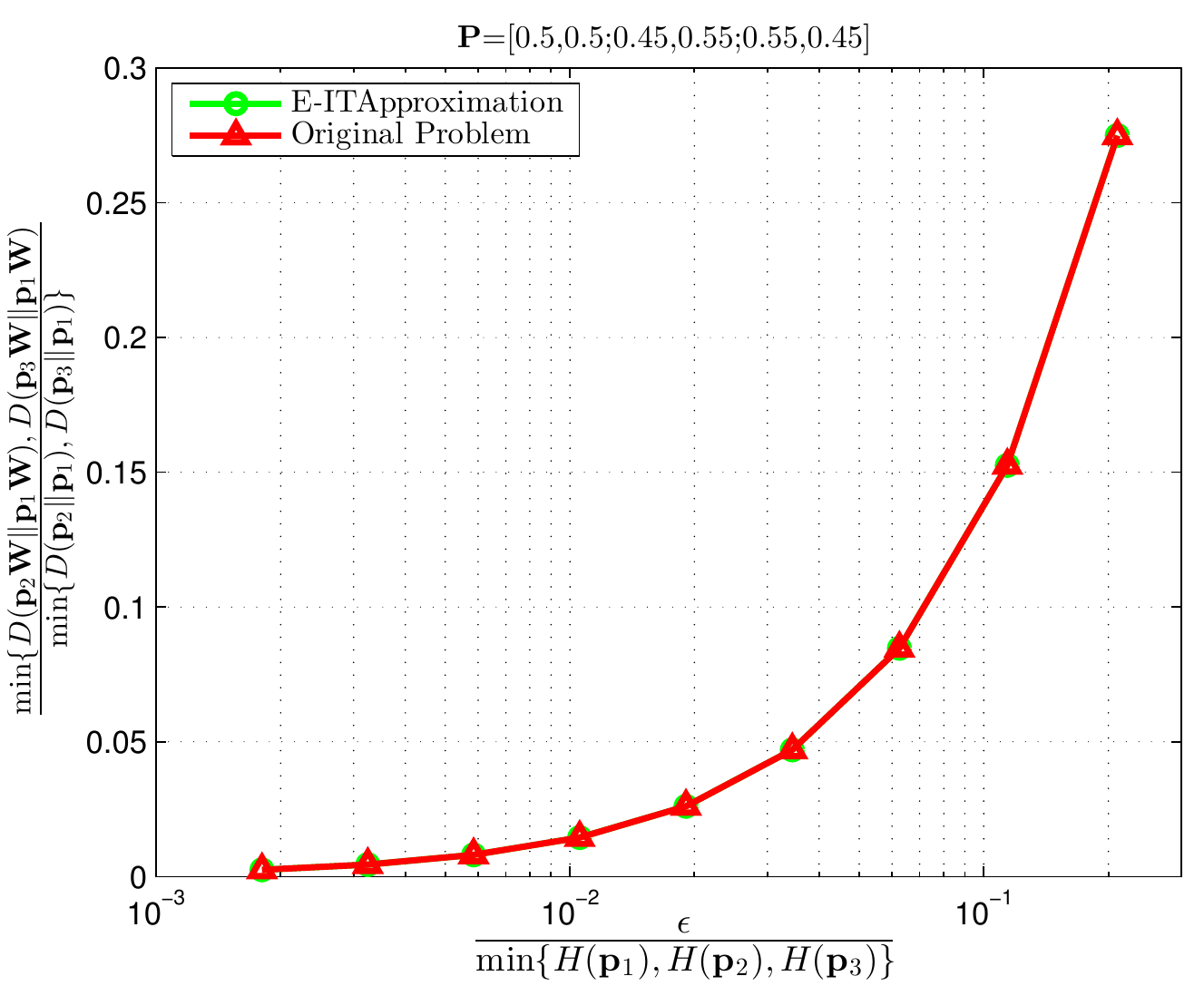}
		\caption{Triple 1}
		\label{subfig:MHTbs_1}
	\end{subfigure}%
	\begin{subfigure}{0.5\textwidth}
		\centering
		\includegraphics[width=.79\columnwidth]{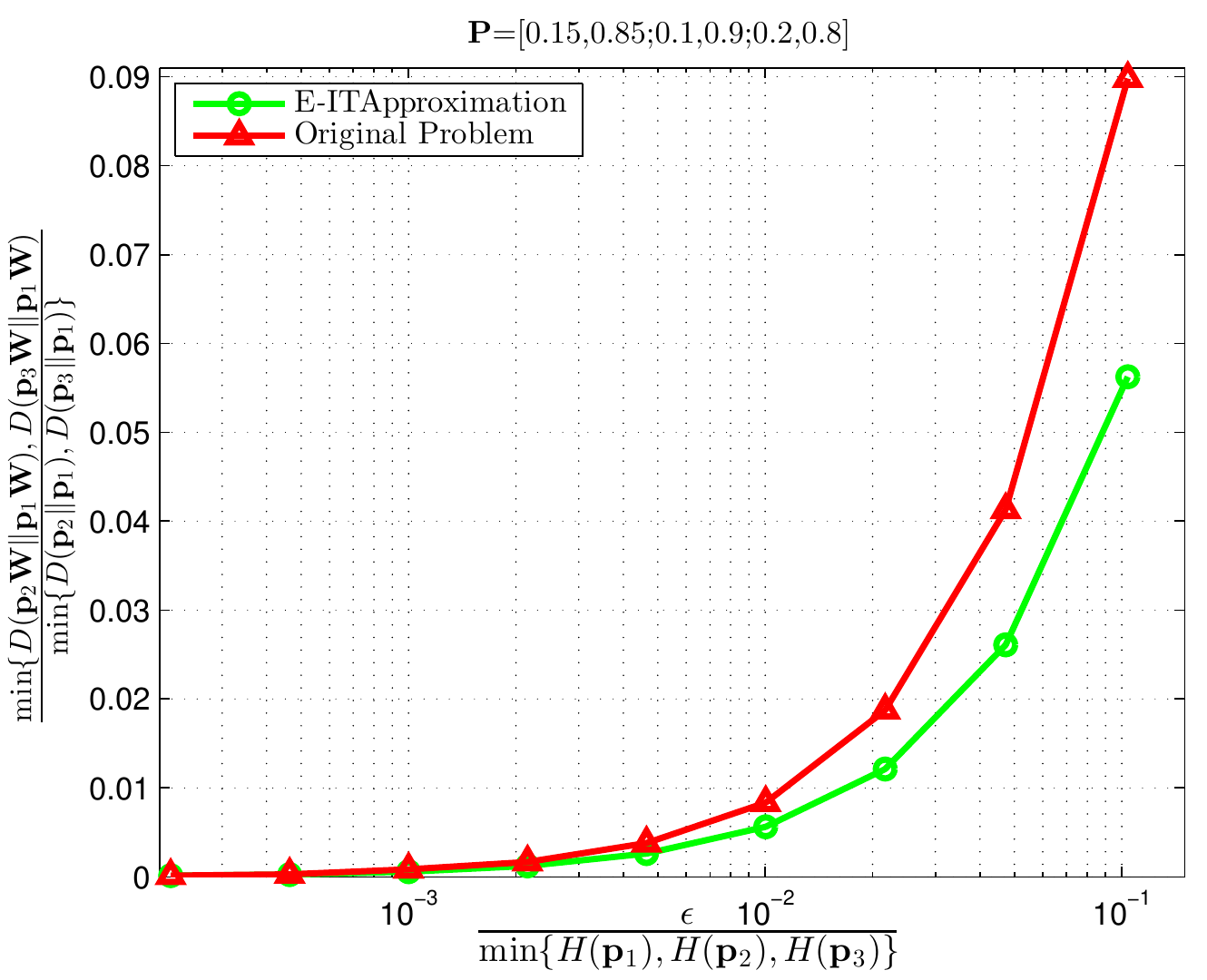}
		\caption{Triple 2}
		\label{subfig:MHTbs_4}
	\end{subfigure} \\
	\caption{The relative utilities of $\mathbf{W}'$ and $\mathbf{W}^*$ for the two distribution tuples in Table \ref{Table:MHT_3binarySources}}
	\label{fig:MHTbs}
\end{figure*}

\section{Concluding Remarks}\label{section:conclusion}
We have systematically studied the problem of publishing large datasets for binary and $m$-ary hypothesis testing under privacy constraints. Our goal, broadly, is to characterize the guarantees that can be made on the error exponents when a MI-based constraint on the leakage of data from any source class  is bounded. Our model seeks to understand if one can find the true probability distribution of a given dataset among a set of possible distributions without revealing the respondents of the data. We have shown that the optimal  PUT  is achieved through a randomizing privacy mechanism which maximizes the minimum of a set of the relative entropies  between pairs of distributions (one for each hypothesis), while ensuring that the MI-based leakages for all source classes are bounded. Focusing on the high privacy regime, we have developed an E-IT approximation of the PUT problem. For this problem, we have shown that the optimal mechanism can be viewed as a perturbation of a perfect privacy mechanism where the perturbation is computed as a solution to a convex optimization problem. As is expected of statistical metrics such as relative entropy and MI, our results reveal that the randomizing mechanism perturbs the statistical outliers the most for each source class. Such a mechanism ensures both utility (predominantly provided by the non-outliers) while preserving privacy of those most vulnerable to inference attacks.

Future work includes developing optimal mechanisms and PUTs for all possible leakage levels. Comparison with other privacy metrics is also of interest.


%

\appendices

\section{Proof of Proposition \ref{Proposition:approximation_MIandRE}}\label{proof:Proposition_approximation_MIandRE}
\begin{proof}
	Consider the high privacy regime in which $0\leq\epsilon_k \ll \min\{H(\mathbf{p}_k),k\in\{1,..,m\}\},\, k\in\{1,..,m\}$. In this regime, $\mathbf{W}$ can be written as a perturbation of $\mathbf{W}_0$ via
	\begin{align}
	\label{eq:Theta_to_W}
	\mathbf{W}=\mathbf{W}_0+\boldsymbol{\Theta}
	\end{align}
	where $\mathbf{W}_0$ is a mechanism achieving perfect privacy with all rows equal to $\mathbf{w}_0$, where $\mathbf{w}_0$ is chosen such that its entries $w_{0j}\neq 0, \forall j \in \{1,\ldots,N\}$, and $\mathbf{\Theta}$ is a matrix with
	\begin{align}
	\label{eq:Con_Theta1}
	\sum_{j=1}^{N}\Theta_{ij}&=0,   \qquad \,\,\forall\,  i\in \{1, \ldots, M\} \\*
	\label{eq:Con_Theta2}
	|\Theta_{ij}&| \leq \rho w_{0j}  \quad \forall\,  i\in \{1, \ldots, M\} ,j\in \{1, \ldots, N\}.
	\end{align}
	where the radius of the neighborhood around $\mathbf{w}_0$ is~$\rho \!\in\! [0,1)$.
	
	Note that \eqref{eq:Con_Theta1} is derived from the row stochasticity of $\mathbf{W}$ and $\mathbf{W}_0$. The constraint in \eqref{eq:Con_Theta2} captures the fact that approximating about a perfect privacy achieving mechanism   requires restricting the entries of the perturbation matrix $\boldsymbol{\Theta}$ to be within a fraction $\rho$ of~$\mathbf{w}_0$.
	
	The perturbation modeled in \eqref{eq:Theta_to_W}-\eqref{eq:Con_Theta2} implies that every row in $\mathbf{W}$ and the output distribution $\mathbf{p}_1\mathbf{W}$ and $\mathbf{p}_k\mathbf{W}$ for all $k\in\{2,\ldots,m\}$ are in a neighborhood about $\mathbf{w}_0$ given by
	\begin{subequations}
		\begin{align}
		\label{eq:OrigProbHP_ApproxCond3}
		&\big|W_{ij}-w_{0j}\big|\leq\rho w_{0j}, \\
		\label{eq:OrigProbHP_ApproxCond1}
		&\big|\big(\mathbf{p}_k\mathbf{W}\big)_j-w_{0j}\big|\leq\rho w_{0j}, \text{ for all } k \in\{1,\ldots,m\}
		\end{align}
	\end{subequations}
	for all $i \in \{1, \ldots, M\}$ and $j \in\{1, \ldots, N\}$.  
	In this neighborhood, we can approximate the relative entropy $D(\mathbf{p}_k\mathbf{W} \| \mathbf{p}_1\mathbf{W})$ using a Taylor series around $\mathbf{W}_0$ as
	\begin{align}
	  D(\mathbf{p}_k\mathbf{W} \| \mathbf{p}_1\mathbf{W}) 
 &= \frac{1}{2}\sum_{j=1}^{N}\frac{\big(\sum_{i=1}^{M}p_{ki}\Theta_{ij}-\sum_{i=1}^{M}p_{1i}\Theta_{ij}\big)^2}{\sum_{i=1}^{M}p_{ki}W_{ij}}\nonumber\\
	&\quad +o\Big(\big\|(\mathbf{p}_k-\mathbf{p}_1)\boldsymbol{\Theta}\big[(\mathbf{p}_k\mathbf{W})^{-\frac{1}{2}}\big]\big\|^2_\infty\Big)\\
	\label{eq:ApproximateD_step2}
	&\approx   \frac{1}{2}\Big\|(\mathbf{p}_k-\mathbf{p}_1)\boldsymbol{\Theta}\big[(\mathbf{p}_k\mathbf{W})^{-\frac{1}{2}}\big]\Big\|^2\\
	\label{eq:ApproximateD_step3}
	&\approx  \frac{1}{2}\big\|(\mathbf{p}_k-\mathbf{p}_1)\boldsymbol{\Theta}[(\mathbf{w}_0)^{-\frac{1}{2}}]\big\|^2
	\end{align}
	where \eqref{eq:ApproximateD_step2} results from approximating the relative entropy by the $\chi^2$-divergence   \cite[Theorem 4.1]{ITSt_Tutorial}, and \eqref{eq:ApproximateD_step3} results from applying the neighborhood condition in \eqref{eq:OrigProbHP_ApproxCond1}. 
	
	Similarly, one can approximate the MI between the source class $k\in\{1,\ldots,m\}$ and its output as 
	\begin{align}
	  I(\mathbf{p}_k,\mathbf{W}) 
	\label{eq:ApproximateI_step1}
	 &=\frac{1}{2}\sum_{i=1}^{M}p_{ki}\bigg(\sum_{j=1}^{N}\frac{\big(W_{ij}-\sum_{i=1}^{M}p_{ki}W_{ij}\big)^2}{W_{ij}} \nonumber\\*
	&\qquad  + o\Big(\big\|\big(\mathbf{W}_i-\mathbf{p}_k\mathbf{W}\big)\big[(\mathbf{W}_i)^{-\frac{1}{2}}\big]\big\|^2_\infty\Big)\bigg)\\\
	\label{eq:ApproximateI_step2}
	&\approx  \frac{1}{2}\sum_{i=1}^{M}p_{ki}\Big\|\big(\mathbf{W}_i-\mathbf{p}_k\mathbf{W}\big)\big[(\mathbf{W}_i)^{-\frac{1}{2}}\big]\Big\|^2 \\
	\label{eq:ApproximateI_step3}
	&\approx \frac{1}{2}\sum_{i=1}^{M}p_{ki}\big\|\big(\mathbf{w}_0\!+\!\boldsymbol{\Theta}_{i}\!-\!\mathbf{p}_k\mathbf{W}\big)[(\mathbf{w}_0)^{-\frac{1}{2}}]\big\|^2 \\
	\label{eq:ApproximateI_step4}
	&\approx \frac{1}{2}\sum_{i=1}^{M}p_{ki}\big\|\boldsymbol{\Theta}_{i}[(\mathbf{w}_0)^{-\frac{1}{2}}]\big\|^2,
	\end{align}
	where $\mathbf{W}_{i}$ and $\boldsymbol{\Theta}_{i}$ are the $i^{\mathrm{th}}$ rows of $\mathbf{W}$ and $\boldsymbol{\Theta}$, respectively.  Let $\mathbf{A}$ be a matrix with entries
	\begin{align}
	\label{eq:A_to_Theta}
	A_{ij}=\frac{\Theta_{ij}}{\sqrt{w_{0j}}},\qquad  i\in \{1, \ldots, M\} ,j \in \{1, \ldots, N\} ,
	\end{align} 
	From \eqref{eq:Con_Theta1}, $\mathbf{A}(\sqrt{\mathbf{w}_0})^T=\mathbf{0}$, where $\sqrt{\mathbf{w}_0}$ is the vector whose entries are the square root of the entries of $\mathbf{w}_0$. Thus,  the approximations in~\eqref{eq:EITapproximation_RE} and \eqref{eq:EITapproximation_MI} lead  to~\eqref{eq:EITapproximation_RE_A} and~\eqref{eq:EITapproximation_MI_A} resp.
\end{proof}


\section{Proof of Theorem \ref{Theorem:BHT_approximation_simplification}}

\begin{proof}\label{proof:BHT_approximation_simplification}
	Consider the following optimization problem obtained from \eqref{eq:BHT_approximation} without the constraint $\mathbf{A}(\sqrt{\mathbf{w}_0})^T=\mathbf{0}$:
	\begin{equation}\label{eq:BHT_approximation_a}
	\begin{aligned}
	\max_{\substack{\mathbf{A}}}\quad & \sum_{i,j=1}^{M}\frac{1}{2}(\mathbf{p}_2-\mathbf{p}_1)_i(\mathbf{p}_2-\mathbf{p}_1)_j\mathbf{A}_i\mathbf{A}_j^T\\
	\mathrm{s.t.} \quad &\frac{1}{2}\sum_{i=1}^{M}p_{ki}\mathbf{A}_i\mathbf{A}_i^T\leq \epsilon_k \quad k=1,2\\
	\end{aligned}
	\end{equation}	
	Let $\Delta_i \triangleq (\mathbf{p}_1-\mathbf{p}_2)_i$, $i\in\{1,\ldots,M\}$. Furthermore, let $\mathbf{a}$ be a row vector with entries $a_i$ for all $i$ that $|a_i|$ is the Euclidean norm of the $i^{\mathrm{th}}$ row $\mathbf{A}_i$ of $\mathbf{A}$. Let $\boldsymbol{\Omega}$ denote the symmetric matrix of the cosines of angles between the rows of $\mathbf{A}$, such that its entries $\Omega_{ij}\triangleq\cos\angle(\mathbf{A}_i,\mathbf{A}_j)$, $i,j\in \{1,\ldots,M\}$, with $|\Omega_{ij}|\leq 1$ for $i\neq j$ and $\Omega_{ij}=1$ for $i=j$. Rewriting \eqref{eq:BHT_approximation_a} with these variables, we have
	\begin{equation}\label{eq:Vari_Mat2Vec1}
	\begin{aligned}
	\max_{\substack{\mathbf{a},\boldsymbol{\Omega}}}\quad & \sum_{i=1}^{M}\sum_{j=1}^{M}\frac{1}{2}\Delta_i\Delta_j|a_i \| a_j|\Omega_{ij}\\
	\mathrm{s.t.} \quad 
	&\Omega_{ii}= 1 \\
	&|\Omega_{ij}|\leq 1 \quad i\neq j \in \{1, \ldots, M\} \\
	& \frac{1}{2}\sum_{i=1}^{M}p_{ki}a_i^2\leq \epsilon_k, \quad k=1,2\\
	\end{aligned}
	\end{equation}
	Consider first the optimization over $\boldsymbol{\Omega}$. Since the objective   is linear in $\boldsymbol{\Omega}$ and the feasible region of $\boldsymbol{\Omega}$ is a hypercube, \eqref{eq:Vari_Mat2Vec1} is a linear program whose optimal solution is at one of the extreme points of the hypercube \cite[Theorem 3.5.3]{Nonlinear_Programming}, i.e., the optimal solution $\mathbf{\Omega}^*$ has entries $|\Omega_{i,j}^*|=1$ for all $i,j$. Thus, all the rows of an $\mathbf{A}$ maximizing~\eqref{eq:BHT_approximation_a} are parallel, and therefore, the optimal solution $\mathbf{A}^*$ of~\eqref{eq:BHT_approximation_a} is a rank-1 matrix.
	
	In addition, from the objective function in \eqref{eq:Vari_Mat2Vec1}, if the signs of $\Delta_1\Delta_i$ and $\Delta_1\Delta_j$ are known for any $i,j\in \{2,\ldots, M\}$, the sign of $\Delta_i\Delta_j$ can be determined. Furthermore, maximizing the objective requires that $\Omega_{ij}$ has the same sign as its coefficient $\Delta_i\Delta_j$. Therefore, $\Omega^*_{ij}=\Omega^*_{1i}\Omega^*_{1j}$ for all $i,j\in\{2,\ldots,M\}$, i.e., $\boldsymbol{\Omega}^*$ has only $M-1$ independent entries $\Omega^*_{1j}$, $j\in \{2,\ldots,M\}$ with $\Omega_{ii}^*=1$ for all $i$.
	
	Thus, we see that the optimization in \eqref{eq:Vari_Mat2Vec1} depends on only $M$ values of $|a_i|$, $i\in \{1,\ldots,M\}$, and $M-1$ signs of $\Omega^*_{1i}$, $i\in \{2,\ldots,M\}$. Let $\mathbf{v}$ denote a unit norm vector with no zero entry, 
	and $\theta_i^*$ for all $i$ represent the direction of $i^{\mathrm{th}}$ row of $\mathbf{A}$ with respect to $\mathbf{v}$, such that $\Omega^*_{ij}=\theta_i^*\theta_j^*$, and the $i^{\mathrm{th}}$ row of $\mathbf{A}$ can be written as $\mathbf{A}_i=\theta_i^*|a_i|\mathbf{v}=a_i\mathbf{v}$. 
	The optimization in \eqref{eq:Vari_Mat2Vec1} can now be written as a function of the vector $\mathbf{a}$ as
	\begin{equation}\label{eq:Vari_Mat2Vec2}
	\begin{aligned}
	\max_{\mathbf{a}}\quad & \sum_{i=1}^{M}\sum_{j=1}^{M}\frac{1}{2}(\mathbf{p}_2-\mathbf{p}_1)_i(\mathbf{p}_2-\mathbf{p}_1)_ja_ia_j\\
	\mathrm{s.t.} \quad 
	& \frac{1}{2}\sum_{i=1}^{M}p_{ki}a_i^2\leq \epsilon_k,\quad k = 1,2.
	\end{aligned}
	\end{equation}
	The optimal solution $\mathbf{A}^*$ of \eqref{eq:BHT_approximation_a} is related to $\mathbf{a}^*$ optimizing~\eqref{eq:Vari_Mat2Vec2} as follows:
	\begin{align}
	\label{eq:optimal_A}
	\mathbf{A}^*=(\mathbf{a}^*)^T\mathbf{v}.
	\end{align}
	The optimal solution in \eqref{eq:Vari_Mat2Vec2} yields both the magnitude and sign of $a_i$ for all $i$.
	Also, $\mathbf{v}$ in \eqref{eq:optimal_A} can be chosen to satisfy
	\begin{align}
	\label{eq:optimalA_ORTH}
	\mathbf{A}^*(\sqrt{\mathbf{w}_0})^T=(\mathbf{a}^*)^T\mathbf{v}(\sqrt{\mathbf{w}_0})^T=\mathbf{0}.
	\end{align}
	Thus, by using \eqref{eq:optimal_A} and \eqref{eq:optimalA_ORTH} and solving for $\mathbf{a}^*$ in \eqref{eq:BHT_approx_2}, we obtain $\mathbf{A}^*$ in~\eqref{eq:BHT_approximation} as desired. For any $\bw_0\in\mathbb{R}^N$, the condition $N= 2$ is sufficient to obtain a $\bv$ satisfying $\mathbf{v}(\sqrt{\mathbf{w}_0})^T=0$, i.e., $\mathbf{A}^*(\sqrt{\mathbf{w}_0})^T=\mathbf{0}$. Therefore, binary output alphabets suffices.
\end{proof}

\section{Proof of Theorem \ref{Theorem:BHT_optimalsol_HPapprox}}\label{proof:Theorem_BHT_optimalsol_HPapprox}
To prove Theorem \ref{Theorem:BHT_optimalsol_HPapprox}, we use Theorem~\ref{Theorem:BHT_approximation_simplification} and  two lemmas. The problem in \eqref{eq:BHT_approx_2} maximizes a convex function over a convex set, and thus, it is not a convex program. However, we show how the problem can be reduced to a convex program  and we also obtain a closed-form  solution. 

\begin{lemma}\label{lemma_equivalentOS}
	The following convex program completely determines the solutions of \eqref{eq:BHT_approx_2}, 
	\begin{equation}\label{eq:BHT_approx_3simp}
	\begin{aligned}
	\max_{\substack{\mathbf{a}}}\quad & 
	\frac{1}{2}\lambda_{\mathrm{p}}\mathbf{a}(\mathbf{v}_{\mathrm{p}})^T\\
	\mathrm{s.t.} \quad &\frac{1}{2}\mathbf{a}\big[\mathbf{p}_k\big]\mathbf{a}^T\leq \epsilon_k,\quad k=1,2
	\end{aligned}
	\end{equation}
	such that the optimal solutions of \eqref{eq:BHT_approx_2} are $\pm\mathbf{a}^*$ where $\mathbf{a}^*$ is the optimal solution of \eqref{eq:BHT_approx_3simp}.
\end{lemma}
\begin{proof}
	For the optimization problem in \eqref{eq:BHT_approx_2}, the matrix $(\mathbf{p}_2-\mathbf{p}_1)^T(\mathbf{p}_2-\mathbf{p}_1)$ is rank-1 with eigenvalue $\lambda_{\mathrm{p}}=\|\mathbf{p}_2-\mathbf{p}_1\|^2$ and eigenvector $\mathbf{v}_{\mathrm{p}}=\frac{\mathbf{p}_2-\mathbf{p}_1}{\|\mathbf{p}_2-\mathbf{p}_1\|}$.
	Thus, we have
	\begin{align}
	\frac{1}{2}\mathbf{a}(\mathbf{p}_2-\mathbf{p}_1)^T(\mathbf{p}_2-\mathbf{p}_1)\mathbf{a}^T=\frac{1}{2}\lambda_{\mathrm{p}}(\mathbf{a}(\mathbf{v}_{\mathrm{p}})^T)^2,
	\end{align} 
	leading to the following optimization problem 
	\begin{equation}\label{eq:BHT_approx_3}
	\begin{aligned}
	\max_{\substack{\mathbf{a}}}\quad & 
	\frac{1}{2}\lambda_{\mathrm{p}}(\mathbf{a}(\mathbf{v}_{\mathrm{p}})^T)^2\\
	\mathrm{s.t.} \quad &\frac{1}{2}\mathbf{a}\big[\mathbf{p}_k\big]\mathbf{a}^T \leq \epsilon_k, \quad k = 1,2.
	\end{aligned}
	\end{equation}
	In \eqref{eq:BHT_approx_3simp} and \eqref{eq:BHT_approx_3}, the two objectives depend  on $\mathbf{a}$ in the same manner and their constraint functions are the same. Hence  the optimal solution $\mathbf{a}^*$ of \eqref{eq:BHT_approx_3simp} optimizes \eqref{eq:BHT_approx_3}. Since the objective and constraint functions of \eqref{eq:BHT_approx_3} are even, $-\mathbf{a}^*$ is feasible and yields the optimal value, i.e, $-\mathbf{a}^*$ is also optimal for~\eqref{eq:BHT_approx_3}.
\end{proof}

The optimal solution $\mathbf{a}^*$ of \eqref{eq:BHT_approx_3simp} can be evaluated by observing that at $\mathbf{a}^*$, either one or both constraints are active. The following lemma summarizes the optimal solution of \eqref{eq:BHT_approx_3simp}. From Lemma \ref{lemma_equivalentOS}, one can then obtain the optimal solution \eqref{eq:BHT_approx_2}.
\begin{lemma}\label{lamma_optimalalpha}
	The optimal solutions of \eqref{eq:BHT_approx_2} are given by: 
	\begin{enumerate}
\item if only the first constraint is active, i.e., $\epsilon_1$ and $\epsilon_2$ satisfy \eqref{eq:const1_active},
the optimal solution $\boldsymbol{\alpha}^*$ is \eqref{eq:optalpha_const1};
\item if only the second constraint is active, i.e., $\epsilon_1$ and $\epsilon_2$ satisfy \eqref{eq:const2_active},
the optimal solution $\boldsymbol{\alpha}^*$ is \eqref{eq:optalpha_const2};	
\item when both constraints are active, the optimal solution $\boldsymbol{\alpha}^*$ is \eqref{eq:optimalalpha}
with $\eta_1^*,\eta_2^*>0$ satisfying \eqref{eq:eta1} and~\eqref{eq:eta2}.
	\end{enumerate}	
\end{lemma}
\begin{proof}
	From Lemma \ref{lemma_equivalentOS}, to find the optimal solutions of \eqref{eq:BHT_approx_2}, it suffices to find the optimal solution to \eqref{eq:BHT_approx_3simp}. In \eqref{eq:BHT_approx_3simp}, the objective function is linear in $\mathbf{a}$. Since $\mathbf{p}_1$ and $\mathbf{p}_2$ are interior points of probability simplex, both $\big[\mathbf{p}_1\big]$ and $\big[\mathbf{p}_2\big]$ are positive definite, i.e., the two constraint functions are convex in $\mathbf{a}$. Thus, this is a convex program. In addition, $\mathbf{a}=\mathbf{0}$ strictly satisfies the two constraints for positive $\epsilon_1$ and $\epsilon_2$, which means that \eqref{eq:BHT_approx_3simp} satisfies Slater's condition  \cite[Sec.~5.2.3]{boydconvex}. Therefore, the convex program has zero duality gap, and the optimal solutions are given by the following Karush\--Kuhn\--Tucker (KKT) conditions \cite[Sec.~5.5.3]{boydconvex}:
	\begin{subequations}
		\begin{align}
		\label{eq:KKTcondition_derivative}
		\hspace{-.1in} \nabla\big\{f_0(\mathbf{a}^*)\!+\!\eta_1^*\big(\epsilon_1\!-\!
		f_1(\mathbf{a}^*)\big)&\!+\!\eta_2^*\big(\epsilon_2\!-\!  f_2(\mathbf{a}^*)\big)\big\}  \!=\!0\\
		\label{eq:KKTcondition_loose1}
		\eta_1^*\big(\epsilon_1-f_1(\mathbf{a}^*)\big) &=0\\
		\label{eq:KKTcondition_loose2}
		\eta_2^*\big(\epsilon_2-f_2(\mathbf{a}^*)\big)&=0\displaybreak[0]\\
		\label{eq:KKTcondition_con1}
		f_1(\mathbf{a}^*)& \leq \epsilon_1\\
		\label{eq:KKTcondition_con2}
		f_2(\mathbf{a}^*)& \leq \epsilon_2\\
		\label{eq:KKTcondition_dualcon1}
		\eta_1^*&\geq 0\\*
		\label{eq:KKTcondition_dualcon2}
		\eta_2^* &\geq 0.
		\end{align}
	\end{subequations}
	where $f_0,f_1$, and $f_2$ represent the objective and two constraint functions of \eqref{eq:BHT_approx_3simp}, respectively, $\mathbf{a}^*$ is the optimal solution of~\eqref{eq:BHT_approx_3simp}, and $\eta^*_1$ and $\eta^*_2$ are the optimal solutions of the dual problem of \eqref{eq:BHT_approx_3simp}. From \eqref{eq:KKTcondition_derivative}, we have
	\begin{align}
	\label{eq:optimal_alpha_inproof}
	\mathbf{a}^*=\frac{\lambda_{\mathrm{p}}}{2}\mathbf{v}_{\mathrm{p}}\Big(\eta_1^*\big[\mathbf{p}_1\big]+\eta_2^*\big[\mathbf{p}_2\big]\Big)^{-1}.
	\end{align}
	When $\eta^*_1>0$ and $\eta^*_2 = 0$, i.e., the first constraint is active, the optimal solution $\mathbf{a}^*$ of \eqref{eq:BHT_approx_3simp} is
	\begin{align}
	\label{eq:oneconstraint_proof_alpha1}
	\mathbf{a}^*=\frac{\lambda_{\mathrm{p}}}{2}\mathbf{v}_{\mathrm{p}}\big[(\eta_1^*\mathbf{p}_1)^{-1}\big],
	\end{align}
	such that from \eqref{eq:KKTcondition_loose1}, 
	\begin{align}
	\label{eq:oneconstraint_proof_eta1}
	\eta_1^*=\sqrt{\frac{(\lambda_{\mathrm{p}})^2(\mathbf{v}_{\mathrm{p}})^T\big[(\mathbf{p}_1)^{-1}\big]\mathbf{v}_{\mathrm{p}}}{8\epsilon_1}}.
	\end{align}
	Substituting $\eta_1^*$ from \eqref{eq:oneconstraint_proof_eta1} in \eqref{eq:oneconstraint_proof_alpha1}, we obtain 
	\begin{align}
	\label{eq:optalpha_const1_p}
	\mathbf{a}^*=\sqrt{\frac{2\epsilon_1}{\mathbf{v}_{\mathrm{p}}\big[(\mathbf{p}_1)^{-1}\big](\mathbf{v}_{\mathrm{p}})^T}}\mathbf{v}_{\mathrm{p}}\big[(\mathbf{p}_1)^{-1}\big]
	\end{align}
	In addition, for $\eta^*_2 = 0$, \eqref{eq:KKTcondition_con2} is a strict inequality if and only if   $\epsilon_1$ and $\epsilon_2$ satisfy \eqref{eq:const1_active}.
	
	When $\eta^*_1=0$ and $\eta^*_2 > 0$, with the same deduction based on \eqref{eq:KKTcondition_loose2} and \eqref{eq:optimal_alpha_inproof}, the optimal solution $\mathbf{a}^*$ is
	\begin{align}
	\label{eq:optalpha_const2_p}
	\mathbf{a}^*=\sqrt{\frac{2\epsilon_2}{\mathbf{v}_{\mathrm{p}}\big[(\mathbf{p}_2)^{-1}\big](\mathbf{v}_{\mathrm{p}})^T}}\mathbf{v}_{\mathrm{p}}\big[(\mathbf{p}_2)^{-1}\big].
	\end{align} and \eqref{eq:KKTcondition_con1} is a strict  inequality if and only if   $\epsilon_1$ and $\epsilon_2$ satisfy~\eqref{eq:const2_active}. 
	
	When $\eta^*_1>0$ and $\eta^*_2 > 0$, the optimal solution $\mathbf{a}^*$ is given by \eqref{eq:optimal_alpha_inproof}. From \eqref{eq:KKTcondition_loose1} and \eqref{eq:KKTcondition_loose2}, we have $f_1(\mathbf{a}^*)=\epsilon_1$ and $f_2(\mathbf{a}^*)=\epsilon_2$, such that $\eta_1^*$ and $\eta_2^*$ satisfy \eqref{eq:eta1} and \eqref{eq:eta2}.
	
	Finally, since \eqref{eq:optimal_alpha_inproof}, \eqref{eq:optalpha_const1_p} and \eqref{eq:optalpha_const2_p} yield  $\mathbf{a}^*$ for \eqref{eq:BHT_approx_3simp}, the optimal solutions for \eqref{eq:BHT_approx_2} are obtained by considering both solutions $\pm\mathbf{a}^*$ as proved in   Lemma \ref{lemma_equivalentOS}.
\end{proof}

The  proof of Theorem \ref{Theorem:BHT_optimalsol_HPapprox}  follows directly from Theorem \ref{Theorem:BHT_approximation_simplification}, Lemma \ref{lemma_equivalentOS} and Lemma \ref{lamma_optimalalpha} as follows.  

\begin{proof}
	The optimal privacy mechanism $\mathbf{W}'$ is a perturbation of $\mathbf{W}_0$ as $\mathbf{W}'=\mathbf{W}_0+\mathbf{A}^*\big[\sqrt{\mathbf{w}_0}\big]$, where $\mathbf{A}^*$ optimizes \eqref{eq:BHT_approximation}. From Theorem \ref{Theorem:BHT_approximation_simplification}, $\mathbf{A}^*$ is given by $\mathbf{A}^*=(\mathbf{a}^*)^T\mathbf{v}$, where $\mathbf{v}$ is a unit norm $M$-dimensional vector that is orthogonal to $\sqrt{\mathbf{w}_0}$, and $\mathbf{a}^*$ is the optimal solution of \eqref{eq:BHT_approx_2} as presented in Lemma~\ref{lamma_optimalalpha}. Therefore, the optimal privacy mechanism $\mathbf{W}'$ is~\eqref{eq:W_generation}.	
\end{proof}

\section{Proof of Lemma \ref{lemma:Approximation_Renyidiv}}\label{proof:lemma_Approximation_Renyidiv}
\begin{proof}
	From \cite[Theorem~8]{fDivergenceInequalities_Sason}, for $\alpha\in (0,1)$ the order-$\alpha$  Hellinger divergence\footnote{The the order-$\alpha$  Hellinger divergence  $\mathcal{H}_{\alpha}(\bp\|\mathbf{q})  = D_{ f_\alpha }(\bp\|\mathbf{q})$  where $D_f(\bp\|\mathbf{q}) = \sum_i q_i f \big( \frac{p_i}{q_i} \big)$ is the $f$-divergence and $f_\alpha(t) = \frac{t^\alpha-1}{\alpha-1}$. } $\mathcal{H}_{\alpha}$ and  the relative entropy  satisfy
	\begin{align}\label{eq:Inequality_RelativeVSHellinger}
	\kappa_{\alpha}(\beta_2)\leq \frac{D(\bp_2\|\bp_1)}{\mathcal{H}_{\alpha}(\bp_2\|\bp_1)}\leq \kappa_{\alpha}(\beta_1^{-1}).
	\end{align}
	Here, $\kappa_{\alpha}$ is a continuous function defined as
	\begin{align}
	\kappa_{\alpha}(t)=\begin{cases}
	\frac{(1-\alpha)t\log t}{1-t^{\alpha}+\alpha t-\alpha}  &\quad  t\in(0,1)\cup (1,\infty)\\
	\alpha^{-1}\log e   &\quad t=1\\
	\log e &\quad  t=0
	\end{cases},
	\end{align} and   $\beta_1$ and $\beta_2$, which depend on $\bp_1$ and $\bp_2$, are  defined as
	\begin{align}
	\beta_1
	&=\exp\Big(-\log\max_i\frac{p_{2i}}{p_{1i}}\Big)=\bigg(\max_i\frac{p_{2i}}{p_{1i}}\bigg)^{\log \frac{1}{e}}\\
	\beta_2
	&=\exp\Big(-\log\max_i\frac{p_{1i}}{p_{2i}}\Big)=\bigg(\max_i\frac{p_{1i}}{p_{2i}}\bigg)^{\log\frac{1}{e}}.
	\end{align}
	Assuming $\bp_1\ne\bp_2$, we know that $\max_i\frac{p_{1i}}{p_{2i}},\max_i\frac{p_{2i}}{p_{1i}}>1$. Since $\log\frac{1}{e}<0$,  $\beta_1,\beta_2<1$. However, it is also  clear that 
\begin{equation}
\bp_2\to\bp_1\;\Longrightarrow\;\beta_j \to 1, \quad \forall\, j = 1,2.  \label{eqn:cont_beta}
	\end{equation}	
	Since $\kappa_\alpha$ is continuous and $\kappa_\alpha(1) = \alpha^{-1}\log e$, from~\eqref{eqn:cont_beta},
	\begin{align}
\bp_2\to\bp_1 \;  \Longrightarrow\;		 \kappa_\alpha(\beta_2) \to\frac{\log e}{\alpha} , \kappa_\alpha(\beta_1^{-1}) \to\frac{\log e}{\alpha}.
	\end{align}
	Consequently, from \eqref{eq:Inequality_RelativeVSHellinger},
	\begin{equation}
\bp_2\to\bp_1 \;  \Longrightarrow\;	 \frac{D(\bp_2\|\bp_1)}{\mathcal{H}_\alpha(\bp_2\|\bp_1)} \to \frac{\log e}{\alpha} . \label{eq:HelliDiver_KLDiver}
	\end{equation}

	In addition, from~\cite{fDivergenceInequalities_Sason}, we know that the order-$\alpha$ R\'{e}nyi divergence and  the order-$\alpha$ Hellinger divergence admit the following relationship  for $\alpha\in (0,1)\cup (1,\infty)$:
	\begin{equation}\label{eq:RenyiDiver_KLDiver}
	D_{\alpha}(\bp_2\|\bp_1)=\frac{1}{\alpha-1}\log\big(1+(\alpha-1)\mathcal{H}_{\alpha}(\bp_2\|\bp_1)\big) 
	\end{equation}
	The proof  of  \eqref{eq:approx_RenyiDiverKLDiver} is completed by uniting 
\eqref{eq:HelliDiver_KLDiver} and \eqref{eq:RenyiDiver_KLDiver}. 
\end{proof}

\section{Proof of Lemma \ref{Lemma:MHT_approximation_simplerSDP}}\label{proof:lemma_MHT_approximation_simplerSDP}
\begin{proof}
We first replace the matrix variable $\bA$ of \eqref{eq:MHT_approximation} by a positive semi-definite symmetric matrix $\bB=\bA\bA^T$.  The objective function can thus be rewritten as
\begin{align}
\frac{1}{2}\|(\bp_k-\bp_1)\bA\|^2 &= \frac{1}{2}(\bp_k-\bp_1)\bB(\bp_k-\bp_1)^T\\
&= \frac{1}{2}\Tr\big((\bp_k-\bp_1)^T(\bp_k-\bp_1)\bB\big)
\end{align}
for $k\in\{2,\ldots,m\}$. For $k\in\{1,\ldots,m\}$ the constraint functions are 
\begin{align}
\frac{1}{2}\sum_{i=1}^{M}p_{ki}\|\bA_i\|^2=\frac{1}{2}\sum_{i=1}^{M}p_{ki}\bB_{ii}=\frac{1}{2}\Tr\big([\bp_k]\bB\big) 
\end{align}
Using an auxiliary scalar variable $t$ to represent a common lower bound of the terms $\frac{1}{2}\|(\bp_k-\bp_1)\bA\|^2 $ in the objective function, the problem \eqref{eq:MHT_approximation} without the constraint \eqref{eq:MHT_approximation_A} can be seen to be  equivalent  to~\eqref{eq:MHT_approximation_simplerSDP}.

Define the inner product between matrices $\mathbf{X}$ and $\mathbf{Y}$  as
$
\mathbf{X}\bullet\mathbf{Y}\triangleq \sum_{i=1}^{M}\sum_{j=1}^{M} X_{ij}Y{ij}=\Tr(\mathbf{X}^T\mathbf{Y}).
$ 
Then, the optimization problem in \eqref{eq:MHT_approximation_simplerSDP} can be rewritten as
\begin{align}\label{eq:MHT_approximation_simplerSDPinproof}
-\min_{\substack{\bB, t}} \quad &  -t \nonumber\\
\mathrm{s.t.} \quad 
&\begin{bmatrix}-t& \mathbf{0}\\\mathbf{0}^T&\frac{1}{2}\mathbf{P}_k\end{bmatrix}
\bullet
\begin{bmatrix} 1& \mathbf{0}\\ \mathbf{0}^T&\bB \end{bmatrix}\geq 0  \;\;\; \;\,\forall\,  k\in\{2,\ldots,m\}\nonumber\\
& \begin{bmatrix}\epsilon_k& \mathbf{0}\\\mathbf{0}^T&-[\bp_k]\\\end{bmatrix}
\bullet
\begin{bmatrix}1& \mathbf{0}\\\mathbf{0}^T&\bB\\\end{bmatrix}\geq 0 \;\; \;\forall\, k\in\{1,\ldots,m\}\nonumber\\*
&\begin{bmatrix}1& \mathbf{0}\\\mathbf{0}&\bB\\\end{bmatrix}\succeq 0
\end{align}
where the $M\times M$ matrix $\mathbf{P}_k\triangleq(\bp_k-\bp_1)^T(\bp_k-\bp_1)$.

Since $\mathbf{P}_k$  (for $k\in\{2,\ldots,m\}$)  and $[\bp_k]$  (for $k\in\{1,\ldots,m\}$) are symmetric positive semi-definite matrices, by referring to \cite[Section 8.2]{SDP_Robert}, we know that \eqref{eq:MHT_approximation_simplerSDPinproof} is an SDP.	
\end{proof}
\vspace{-.2in}

\ifCLASSOPTIONcaptionsoff
  \newpage
\fi



%
\bibliographystyle{IEEEtran}
\bibliography{HypothesisTestingHighPrivacy2}

\begin{thebibliography}{10}
\providecommand{\url}[1]{#1}
\csname url@samestyle\endcsname
\providecommand{\newblock}{\relax}
\providecommand{\bibinfo}[2]{#2}
\providecommand{\BIBentrySTDinterwordspacing}{\spaceskip=0pt\relax}
\providecommand{\BIBentryALTinterwordstretchfactor}{4}
\providecommand{\BIBentryALTinterwordspacing}{\spaceskip=\fontdimen2\font plus
\BIBentryALTinterwordstretchfactor\fontdimen3\font minus
  \fontdimen4\font\relax}
\providecommand{\BIBforeignlanguage}[2]{{%
\expandafter\ifx\csname l@#1\endcsname\relax
\typeout{** WARNING: IEEEtran.bst: No hyphenation pattern has been}%
\typeout{** loaded for the language `#1'. Using the pattern for}%
\typeout{** the default language instead.}%
\else
\language=\csname l@#1\endcsname
\fi
#2}}
\providecommand{\BIBdecl}{\relax}
\BIBdecl

\bibitem{economist_article}
\BIBentryALTinterwordspacing
{The Economist Online}, ``The 99 percent,'' 2011. [Online]. Available:
  \url{http://www.economist.com/blogs/dailychart/2011/10/income-inequality-america}
\BIBentrySTDinterwordspacing

\bibitem{PoorBook}
H.~V. Poor, \emph{An Introduction to Signal Detection and Estimation},
  2nd~ed.\hskip 1em plus 0.5em minus 0.4em\relax Springer Texts in Electrical
  Engineering, 1998.

\bibitem{IT_Cover}
T.~M. Cover and J.~A. Thomas, \emph{Elements of Information Theory},
  2nd~ed.\hskip 1em plus 0.5em minus 0.4em\relax Wiley-Interscience, 2006.

\bibitem{EEHT_Tuncel}
E.~Tuncel, ``On error exponents in hypothesis testing,'' \emph{IEEE Trans. on
  Inform. Th.}, vol.~51, no.~8, pp. 2945--2950, 2005.

\bibitem{du_pin_calmon_privacy_2012}
F.~du~Pin~Calmon and N.~Fawaz, ``Privacy against statistical inference,'' in
  \emph{50th {Annual} {Allerton} {Conference} on {Communication}, {Control},
  and {Computing}}, 2012.

\bibitem{RebolloMonedero_TClosenessLikePrivacy_2010}
D.~Rebollo-Monedero, J.~Forne, and J.~Domingo-Ferrer, ``From t-closeness-like
  privacy to postrandomization via information theory,'' \emph{IEEE
  Transactions on Knowledge and Data Engineering}, vol.~22, no.~11, pp.
  1623--1636, 2010.

\bibitem{EITzheng2008}
S.~Borade and L.~Zheng, ``Euclidean information theory,'' in \emph{IEEE
  International Zurich Seminar on Communications}, 2008.

\bibitem{EIT2015}
S.~Huang, C.~Suh, and L.~Zheng, ``Euclidean information theory of networks,''
  \emph{IEEE Trans. on Inform. Th.}, vol.~61, no.~12, pp. 6795--6814, 2015.

\bibitem{boydconvex}
S.~Boyd and L.~Vandenberghe, \emph{Convex optimization}.\hskip 1em plus 0.5em
  minus 0.4em\relax Cambridge university press, 2014.

\bibitem{Liao_Allerton16}
J.~Liao, L.~Sankar, V.~Y.~F. Tan, and F.~du~Pin~Calmon, ``Hypothesis testing in
  the high privacy limit,'' in \emph{54st {Annual} {Allerton} {Conference} on
  {Communication}, {Control}, and {Computing}}, 2016.

\bibitem{Kairouz2014}
P.~Kairouz, S.~Oh, and P.~Viswanath, ``Extremal mechanisms for local
  differential privacy,'' in \emph{Advances in Neural Information Processing
  Systems}, 2014.

\bibitem{DP_ChiSquared_Gaboardi}
M.~Gaboardi, R.~Rogers, and S.~Vadhan, ``Differentially private chi-squared
  hypothesis testing: Goodness of fit and independence testing,'' \emph{\tt
  arXiv:1602.03090 [math.ST]}, 2016.

\bibitem{montazeri_defining_2016}
Z.~Montazeri, A.~Houmansadr, and H.~Pishro-Nik, ``Defining perfect location
  privacy using anonymization,'' in \emph{{Annual} {Conference} on
  {Information} {Science} and {Systems}}, 2016.

\bibitem{montazeri_achieving_2016}
Z.~Montazeri, A.~Houmansadr, and H.~Pishro-Nik, ``Achieving perfect location privacy in {Markov} models using
  anonymization,'' in \emph{{International} {Symposium} on {Information}
  {Theory} and {Its} {Applications}}, 2016.

\bibitem{li_privacy_2015}
Z.~Li and T.~J. Oechtering, ``Privacy on hypothesis testing in smart grids,''
  in \emph{{IEEE} {Information} {Theory} {Workshop}}, 2015.

\bibitem{sankar_utility-privacy_2013}
L.~Sankar, S.~R. Rajagopalan, and H.~V. Poor, ``Utility-privacy tradeoffs in
  databases: {An} information-theoretic approach,'' \emph{IEEE Trans. on
  Inform. For. and Sec.}, vol.~8, no.~6, pp. 838--852, 2013.

\bibitem{sankar_smart_2013}
L.~Sankar, S.~R. Rajagopalan, and S.~Mohajer, ``Smart meter privacy: {A}
  theoretical framework,'' \emph{IEEE Transactions on Smart Grid}, vol.~4,
  no.~2, pp. 837--846, 2013.

\bibitem{salamatian_how_2013}
S.~Salamatian, A.~Zhang, F.~du~Pin~Calmon, S.~Bhamidipati, N.~Fawaz, B.~Kveton,
  P.~Oliveira, and N.~Taft, ``How to hide the elephant-or the donkey-in the
  room: {Practical} privacy against statistical inference for large data.'' in
  \emph{{IEEE} {Global} {Conference} on {Signal} and {Information}
  {Processing}}, 2013.

\bibitem{calmon_fundamental_2015}
F.~P. Calmon, A.~Makhdoumi, and M.~M{\'e}dard, ``Fundamental limits of perfect
  privacy,'' in \emph{{IEEE} {International} {Symposium} on {Information}
  {Theory}}, 2015.

\bibitem{rebollo-monedero_t-closeness-like_2010}
D.~Rebollo-Monedero, J.~Forne, and J.~Domingo-Ferrer, ``From t-closeness-like
  privacy to postrandomization via information theory,'' \emph{IEEE
  Transactions on Knowledge and Data Engineering}, vol.~22, no.~11, pp.
  1623--1636, Nov. 2010.

\bibitem{sankar_information-theoretic_2010}
L.~Sankar, S.~R. Rajagopalan, and H.~V. Poor, ``An information-theoretic
  approach to privacy,'' in \emph{48th {Annual} {Allerton} {Conference} on
  {Communication}, {Control}, and {Computing}}, 2010.

\bibitem{makhdoumi_privacy-utility_2013}
A.~Makhdoumi and N.~Fawaz, ``Privacy-utility tradeoff under statistical
  uncertainty,'' in \emph{51st {Annual} {Allerton} {Conference} on
  {Communication}, {Control}, and {Computing}}, 2013.

\bibitem{calmon_bounds_2013}
F.~P. Calmon, M.~Varia, M.~M{\'e}dard, M.~M. Christiansen, K.~R. Duffy, and
  S.~Tessaro, ``Bounds on inference,'' in \emph{51st {Annual} {Allerton}
  {Conference} on {Communication}, {Control}, and {Computing}}, 2013.

\bibitem{asoodeh_maximal_2015}
S.~Asoodeh, F.~Alajaji, and T.~Linder, ``On maximal correlation, mutual
  information and data privacy,'' in \emph{{IEEE} 14th {Canadian} {Workshop} on
  {Information} {Theory}}, 2015.

\bibitem{asoodeh_privacy-aware_2016}
S.~Asoodeh, F.~Alajaji, and T.~Linder, ``Privacy-aware {MMSE} estimation,'' in \emph{{IEEE} {International}
  {Symposium} on {Information} {Theory}}, 2016.

\bibitem{RedAlert_Nazer}
B.~Nazer, Y.~Y. Shkel, and S.~C. Draper, ``The {AWGN} red alert problem,''
  \emph{IEEE Trans. on Inform. Th.}, vol.~59, no.~4, pp. 2188--2200, 2013.

\bibitem{UnequalErrorProtection_Borade}
S.~Borade, B.~Nakiboglu, and L.~Zheng, ``Unequal error protection: An
  information-theoretic perspective,'' \emph{IEEE Trans. on Inform. Th.},
  vol.~55, no.~12, pp. 5511--5539, 2009.

\bibitem{MinConcave}
K.~L. Hoffman, ``A method for globally minimizing convex functions over convex
  sets,'' \emph{Mathematical Programming}, vol.~20, no.~1, pp. 22--31, 1981.

\bibitem{ITSt_Tutorial}
I.~Csisz{\'a}r and P.~C. Shields, \emph{Information theory and statistics: A
  tutorial}.\hskip 1em plus 0.5em minus 0.4em\relax Now Publishers Inc, 2004.

\bibitem{Tan11}
V.~Y.~F. Tan, A.~Anandkumar, L.~Tong, and A.~S. Willsky, ``A large-deviation
  analysis of the maximum-likelihood learning of {Markov} tree structures,''
  \emph{IEEE Trans. on Inform. Th.}, vol.~57, no.~3, pp. 1714--1735, 2011.

\bibitem{HypothesisTesting_Origin_Blahut}
R.~Blahut, ``Hypothesis testing and information theory,'' \emph{IEEE Trans. on
  Inform. Th.}, vol.~20, no.~4, pp. 405--417, 1974.

\bibitem{Nonlinear_Programming}
M.~S. Bazaraa, H.~D. Sherali, and C.~M. Shetty, \emph{Nonlinear Programming:
  Theory and Algorithms}, 3rd~ed.\hskip 1em plus 0.5em minus 0.4em\relax
  Wiley-Interscience, 2006.

\bibitem{fDivergenceInequalities_Sason}
I.~Sason and S.~Verd{\'u}, ``f-divergence inequalities,'' \emph{IEEE Trans. on
  Inform. Th.}, vol.~62, no.~11, pp. 5973--6006, 2016.

\bibitem{SDP_Robert}
\BIBentryALTinterwordspacing
R.~M. Freund, ``Introduction to semidefinite programming ({SDP}),'' Tech. Rep.,
  2009. [Online]. Available:
  \url{https://ocw.mit.edu/courses/electrical-engineering-and-computer-science/6-251j-introduction-to-mathematical-programming-fall-2009/readings/MIT6\_251JF09\_SDP.pdf}
\BIBentrySTDinterwordspacing

\end{thebibliography}

%

%
%
%




\end{document}